\documentclass[letterpaper, 10 pt, conference,useRomanappendices=false]{ieeeconf}  

\IEEEoverridecommandlockouts                              
\makeatletter

\let\proof\@undefined
\let\endproof\@undefined
\makeatother
\usepackage{cite}
\usepackage{amsmath,amssymb,amsfonts,amsthm,color,float}
\usepackage{algorithmic}
\usepackage{graphicx}
\usepackage{textcomp}
 \usepackage{calc}
 \usepackage{tikz}
\usepackage{cases}
\usepackage{verbatim}
\usepackage{hyperref}
\usepackage{xcolor, soul}
\usepackage{nccmath}

\usetikzlibrary{arrows}
\usetikzlibrary{shapes}
 \usetikzlibrary{calc}
\usetikzlibrary{decorations.pathreplacing}
\usetikzlibrary{arrows,positioning}

\newcommand{\vnorm}[1]{\left|\left|#1\right|\right|}

\DeclareMathOperator*{\argmin}{argmin}

\newcommand{\D}{\mathcal{D}}
\newcommand{\R}{\mathbb{R}}
\renewcommand{\S}{\mathcal{S}}

\newcommand{\change}[1]{#1}

  {\begin{enumerate}{}%
          \item[]%
  }
  {\end{enumerate}}

\theoremstyle{plain}
\newtheorem{thm}{Theorem}
\newtheorem{corollary}{Corollary}
\newtheorem{prop}{Proposition}

\newtheorem{lemma}{Lemma}

\theoremstyle{definition}
\newtheorem{definition}{Definition}

\newtheorem{example}{Example}

\theoremstyle{remark}
\newtheorem{rem}{Remark}

\usepackage{comment}
\specialcomment{shortonly}{}{}

\specialcomment{extendedonly}{}{}

\includecomment{extendedonly}
\excludecomment{shortonly}

\def\BibTeX{{\rm B\kern-.05em{\sc i\kern-.025em b}\kern-.08em
    T\kern-.1667em\lower.7ex\hbox{E}\kern-.125emX}}
\begin{document}
\title{Characterizing Safety: Minimal Control Barrier Functions from Scalar Comparison Systems}

\author{Rohit Konda, Aaron D. Ames, and Samuel Coogan 
\thanks{This research was supported by NSF Awards \# 1544332 and \#1749357.}
\thanks{R. Konda is with the School of Electrical and Computer Engineering, Georgia Institute of Technology, Atlanta, GA 30332. \texttt{rkonda6@gatech.edu}}
\thanks{A. D. Ames is with the Dept. of Mechanical and Civil Engineering, California Institute of Technology, Pasadena, CA. \texttt{ames@caltech.edu}}
\thanks{S. Coogan is with the School of Electrical and Computer Engineering and the School of Civil and Environmental Engineering, Georgia Institute of Technology, Atlanta, GA 30332. \texttt{sam.coogan@gatech.edu}}
}

\maketitle

\begin{abstract}
Verifying set invariance has classical solutions stemming from the seminal work by Nagumo, and defining sets via a smooth barrier function constraint inequality results in computable flow conditions for guaranteeing set invariance.
While a majority of these historic results on set invariance consider flow conditions on the boundary, this paper fully characterizes set invariance through \emph{minimal barrier functions} by directly appealing to a comparison result to define a flow condition over the entire domain of the system. 
A considerable benefit of this approach is the removal of regularity assumptions of the barrier function. This paper also outlines necessary and sufficient conditions for a valid differential inequality condition, giving the minimum conditions for this type of approach.
We also show when minimal barrier functions are necessary and sufficient for set invariance.
\end{abstract}

\section{Introduction}
\label{sec:Introduction}

In the context of dynamical systems, safety has become synonymous with \emph{set invariance}, the property that state trajectories of a system are contained within a given subset of the state space; e.g., see the textbook \cite{blanchini2008set}. Intuitively, invariance can be established by ensuring that a system's vector field evaluated on the boundary of the candidate invariant set is always sub-tangent to the set so that trajectories cannot escape.
The main technical challenge of this approach is in defining an appropriate notion of sub-tangency applicable to general sets and finding conditions that extend over the entire set so they can be used for controller synthesis. Recent work on (control) barrier functions provided conditions for set invariance \cite{ames2014control,ames2017control}, subject to regularity assumptions on the set. The question this paper addresses is: {\it Are these the strongest possible conditions for set invariance?}

The main result of this paper is necessary and sufficient conditions on set invariance that are \emph{minimal} in that they are the least restrictive conditions needed to ensure set invariance.  To obtain this result, we begin considering comparison results for scalar systems which lead to a notion of a minimal solution.  This motivates the introduction of a \emph{minimal barrier function} which leverages a comparison result for scalar systems.  Minimal barrier functions are necessary and sufficient for set invariance and, importantly, they do not require the regularity conditions imposed by the original formulation of barrier functions. Finally, minimal control barrier functions are introduced wherein state dependent input constraints and controller synthesis are considered.    

\change{There is a long and rich history of establishing conditions for set invariance, starting with the seminal work by Nagumo \cite{nagumo1942lage} continuing with Bony \cite{bony1969principe}, Brezis \cite{brezis1970characterization}, and others \cite{ladde1974flow, redheffer1972theorems}.}
The modern literature has predominately focused on \change{extending these classical results to when $\S =\{x:h(x)\geq 0\} $, a subset of the domain $\D \subset \R^n$, is defined by a smooth output function $h: \D \to \R$.}
The most visible example of this is \emph{barrier certificates}, which were first introduced to verify safety properties of hybrid systems in \cite{prajna2004safety}. \change{Directly invoking Nagumo's theorem gives the familiar condition :  $\frac{\partial h}{\partial x}f(x) \geq 0$ on the boundary of $\S$ implies invariance of $\S$, for the state flow $f$.} Extensions of barrier certificates have been plenty, see e.g. \cite{prajna2007framework, wisniewski2015converse}, but a major assumption is \change{regularity of $h$: specifically} that the gradient $\frac{\partial h}{\partial x}$ on the boundary does not vanish and corresponds to the exterior normal vector of $\S$. This assumption is necessary, as a simple counterexample is given by \change{$h(x) = x^3$ and $\dot{x} = -1$. This example is further detailed in Example \ref{ex:counterexample} and also appears in \cite{prajna2004safety}.}

\change{The alternative approach is to enforce a flow constraint over the entire domain: $\frac{\partial h}{\partial x}f(x) \geq -\phi(h(x)) \ \forall x \in \D$ for a scalar function $\phi$. From the conception of the Lyapunov-like flow constraint with $\phi \equiv \mathbf{0}$, as in \cite{prajna2004safety}, significant work has been undertaken to expand the class of functions $\phi$ sufficient to guarantee invariance of $\S$, e.g. see \cite{aubin2009viability, kong2013exponential}. With a view towards obtaining tighter conditions for $\phi$, a new form of (control) barrier functions was recently introduced in \cite{ames2017control}, where $\phi$ is required to be an extended class $\mathcal{K}$ function, \emph{i.e.}, it is strictly increasing and $\phi(0) = 0$.
Importantly, these conditions are necessary and sufficient for set invariance in the case
when $\S$ is compact and $0$ is a regular value of $h$.}

Yet the question remains: can these assumptions, especially with respect to the regularity of $h$, be relaxed further and still guarantee set invariance?
Answering this question is important as it allows for the verification for a larger set of invariance specifications for a given system \change{including significant classes of non-regular sets, such as points, limit cycles, subspaces, etc}.
This leads to the main contribution of the paper: the largest possible set of functions, $\mu$, in which to lower bound the flow via $\frac{\partial h}{\partial x} f(x) \geq  -\mu(h(x)) \ \ \forall x \in \D$.

\begin{shortonly}
Complete proofs of all results are contained in the extended version\footnote{The extended version is available at \texttt{arxiv.org/abs/1908.09323}}. Proofs are presented below for a few main results. The extended version also contains several appendices that extend the basic results here to time-varying systems, \change{present stability results using minimal barrier functions,} compares to other boundary conditions for invariance in the literature, and presents further details regarding necessary conditions to guarantee the existence of a continuous comparison function. 
\end{shortonly}

\begin{extendedonly}
\subsection{Overview of Contributions}

In this paper, we introduce novel conditions for ensuring invariance of a set defined via a smooth barrier function. A major objective of this paper is to characterize conditions, that are, in a certain sense, the minimum conditions required on the resulting differential inequality to certify invariance for a smooth barrier function $h(x)$. The remainder of this paper is organized as follows: 

\vspace{0.1cm}
\noindent \underline{Section \ref{sec:Background}:} introduces standard definitions associated with solutions of autonomous systems and positive invariance. 

Additionally, we present two main theorems for minimal barrier functions.

\begin{itemize}
    \item{
    Theorem \ref{thmbarrier} establishes that $\dot{h} = \frac{\partial h}{\partial x}(x)f(x) \geq -\mu(h(x))$ globally for $x \in \D$ is sufficient for establishing invariance of $\S=\{x:h(x)\geq 0\}$ when $\mu(\cdot)$ is a \emph{minimal function}, that is, solutions of the initial value problem $\dot{w}=-\mu(w)$, $w(0)=0$ remain positive for all time. Notably, $\mu$ need not be Lipschitz continuous. The proof of Theorem \ref{thmbarrier} relies on a comparison result tailored for non-Lipschitz vector fields with potentially nonunique solutions. 
    }
    \item{
    Theorem \ref{thm:minimal} presents necessary and sufficient conditions for verifying that a function $\mu$ is a minimal function. These conditions recover as a special case the instance when $\mu$ is locally Lipschitz.
    }
\end{itemize}

\noindent We also give several examples on the application of minimal barrier functions, including Example \ref{ex:bump}, which delineates an instance when a Lipschitz vector field for the comparison system is impossible. Lastly, we compare minimal barrier functions to prior work in zeroing barrier functions. 

\vspace{0.1cm}
\noindent \underline{Section \ref{sec:Nagumo}:} presents a discussion on regularity and outlines necessary conditions on the existence for minimal functions. More concretely, the main result is:

\begin{itemize}
    \item{
    Under regularity conditions on $\S$, namely, twice-differentiability of $h$, compactness assumptions of level sets of $h$, and the requirement that $\frac{\partial h}{\partial x}(x)\neq 0$ for all $x$ such that $h(x)=0$, Theorem \ref{thm:existloc} proves that a locally Lipschitz minimal function $\mu$ always exists satisfying the barrier function $\frac{\partial h}{\partial x} f(x)\geq -\mu(h(x))$ for all $x$ if $\S$ is invariant. 
    }
\end{itemize}

\vspace{0.1cm}
\noindent \underline{Section \ref{sec:Control}:} considers the case of control systems and controller synthesis through the introduction of minimal control barrier functions.  Systems with state dependent input constraints are considered in this case, and several results are presented on guaranteeing continuity properties of a viable controller. 

\vspace{0.1cm}
\noindent \underline{Section \ref{sec:Conclusion}:} presents concluding remarks.

\vspace{0.1cm}
\noindent \underline{Appendix} also contains discussion on nonautonomous versions of minimal barrier functions, presents certain stability results using MBFs, recalls several classic tangent conditions for verifying invariance and outlines their connection to minimal barrier functions, and discusses conditions for existence of a continuous comparison function.
\end{extendedonly}

\section{Minimal Barrier Functions}
\label{sec:Background}


We study the system
\begin{align}
  \label{eq:sys}
  \dot{x} = f(x)
\end{align}
with state $x\in\D$ where $\D \subseteq \mathbb{R}^n$ is assumed to be an open set and $f : \D \to \mathbb{R}^n$ is assumed to be continuous.
\begin{shortonly}
Under this assumption, for any initial condition $x(0) = x_0$, there exists a maximum time interval of existence $I[x(\cdot)]=[0,\tau_{\max}[x(\cdot)])$ \change{ with $\tau_{\max}[x(\cdot)] > 0$} in which the solution $x(t)$ is guaranteed to exist.
\end{shortonly}

The Lie derivative of $h$ along the vector field $f$ is denoted $L_f h:\D \to \mathbb{R}$ and defined by $L_f h(x):=\frac{\partial h}{\partial x}(x) f(x)$. We denote standard Euclidean norm by $\| \cdot \|$.
\begin{extendedonly}
For a set $\S \subseteq \D$, we refer $\partial \S$ as the boundary of $\S$, $\S^{\circ}$ as the interior of $\S$, and $\overline{\S}$ as the closure of $\S$ with the standard topological definitions.

In general, continuity of $f$ ensures solutions exist for \eqref{eq:sys}, but they need not be unique.
A solution $x(t)$ to \eqref{eq:sys} with the initial condition $x(0)=x_0\in \D$ defined for $t\in [0,\tau)$ is called \emph{maximal} if it cannot be extended for time beyond $\tau$ \cite{hartman1964ordinary}. Given a maximal solution $x(t)$ defined on $[0,\tau)$, we  write $\tau_{\max}[x(\cdot)]$ to denote the right (maximal) endpoint $\tau$ of its interval of existence and we write $I[x(\cdot)]=[0,\tau_{\max}[x(\cdot)])$ to indicate the (maximal) interval of existence of $x(t)$.

To avoid cumbersome notation, we often write simply $\tau_{\max}$ instead of $\tau_{\max}[x(\cdot)]$ when $x(t)$ is clear, \emph{e.g.}, $I[x(\cdot)]=[0,\tau_{\max})$. Further, we write $t\geq 0$ instead of $t\in I[x(\cdot)]$ when clear from context. The system \eqref{eq:sys} is \emph{forward complete} if $\tau_{\max}[x(\cdot)]=\infty$ for all maximal solutions $x(t)$.
\end{extendedonly}

Throughout this paper, we will study invariance of sets defined  as $\S = \{x \in \D : h(x) \geq 0 \}$ for a continuously differentiable function $h : \D \to \mathbb{R}$.

A set $\S \subseteq \D$ is \emph{positively invariant} for \eqref{eq:sys} if, for any $x_0\in \S$, all solutions $x(t)$ with $x(0)=x_0$ satisfy $x(t)\in\S$ for all $t\in I[x(\cdot)]$. 

\begin{extendedonly}
A set $\S\subseteq \D$ is \emph{weakly positive invariant} for \eqref{eq:sys} if, for any $x_0\in \S$, there exists at least one maximal solution $x(t)$ with $x(0)=x_0$ satisfying $x(t)\in \S$ for all $t\in I[x(\cdot)]$. 

If solutions are unique for \eqref{eq:sys} for every $x_0 \in \D$, then the definitions of positive invariance and weak positive invariance coincide. While we are almost exclusively interested in positive invariance in this paper, we will occasionally reference the weaker formulation.
\end{extendedonly}

By studying sets $\S$ defined by inequality constraints of a smooth function $h$, we can develop Lyapunov-like conditions on the time evolution of the scalar value $h$ over the whole domain $\D$. 
In particular, we observe that if $h(x(t)) \geq 0$ can be assured for all $t \geq 0$ and for all initial conditions $x_0 \in \S$, then $\S$ is positively invariant.
In contrast, previous results on barrier functions (for example, {\cite{prajna2004safety}}) focus on extending Nagumo's Theorem, directly \change{verifying that the flow of the system is sub-tangent} to the set $\S$.

\begin{shortonly}
First we define a \emph{minimal solution} $\Tilde{w}(t)$ defined on $[0, \tau)$ for the scalar initial value problem 
\begin{equation}
    \label{eq:IVP}
    \dot{w}=g(w)\text{, } \ w(0) = w_0\in\mathbb{R},
\end{equation}
\change{with $g: \R \to \R$ a continuous function,} as a solution such that, for any other solution $w'(t)$ defined on $[0,\tau)$, $\tilde{w}(t)\leq w'(t)$ for all $t \in [0, \tau)$.

Minimal solutions \change{will be} fundamental in characterizing differential inequalities for systems that do not necessarily have unique solutions. Relaxing uniqueness assumptions is critical, as even systems and barrier functions defined by polynomials can induce nonunique solutions in the differential inequality, as seen in Example \ref{ex:counterexample}. With this, we recall the following differential inequality.
\end{shortonly}

\begin{extendedonly}
First we recall some standard notions of solutions for first order differential equations. 

Consider the initial value problem
\begin{equation}
    \label{eq:IVP}
    \dot{w}=g(w) \quad \quad \quad w(0) = w_0
\end{equation}
where $g:W\to \mathbb{R}$ is a continuous function defined on an open set $W \subseteq \mathbb{R}$ and $w_0\in W$. Again, continuity of $g$ guarantees existence but not uniqueness of solutions to \eqref{eq:IVP}. The following definitions appear in, \emph{e.g.}, \cite [Section 2.2]{Pachpatte:1997oy}. 

\begin{itemize}
    \item 
    A differentiable function $w(t)$ defined on some interval $[0,\tau)$ is a \emph{solution} of \eqref{eq:IVP} if $w(t) \in W$ for $t \in [0,\tau)$, $w(0)=w_0$, and $\dot{w}(t)=g(w(t))$ for all $t \in [0, \tau)$. 
    \item
    A solution $\Tilde{w}(t)$ is a \emph{minimal solution} of \eqref{eq:IVP} on $[0, \tau)$ if, for any other solution $w'(t)$ defined  on $[0,\tau)$, $\tilde{w}(t)\leq w'(t)$ for all $t \in [0, \tau)$.
\end{itemize}

The existence of minimal solutions is guaranteed by the fact that $g(w)$ is continuous on the domain $W$ \cite[Thm 1.3.2]{lakshmikantham1969differential}, while uniqueness of minimal solutions is guaranteed by properties of the standard ordering on $\R$.

Minimal solutions are fundamental for establishing comparison results in scalar differential inequalities. In particular, the following proposition presents a comparison result similar to the one outlined in \cite [Thm 6.3]{bainov2013integral} for which a solution to a differential inequality is bounded below by the minimal solution of the corresponding comparison system. For completeness, a proof is provided that follows closely to the proof of \cite [Thm 6.3]{bainov2013integral}.
\end{extendedonly}

\begin{prop}\cite [Thm 6.3]{bainov2013integral}
  \label{prop:compare}
  Let $\Tilde{w}(t)$ be a minimal solution to the the initial value problem \eqref{eq:IVP} with domain $[0,\tau)$ and $g$ being continuous. If $\eta(t)$ is any differentiable function defined on $[0,\tau)$ such that
    $\dot{\eta}(t)\geq g(\eta(t))\ \text{ for all } t\in [0,\tau)$ and $\eta(0)\geq w_0$, 
  then
    $\eta(t)\geq \Tilde{w}(t) \text{ for all }t\in [0,\tau)$.
\end{prop}

\begin{extendedonly}
\begin{proof}
We initially show $\eta(t) \geq \Tilde{w}(t)$ on any compact time interval $[0, \tau_f]$ with $\tau_f < \tau$. Let $\{\epsilon_n\}$ be a strictly monotonically decreasing sequence with $\lim_{n \to \infty}{\epsilon_n} = 0$ and $\epsilon_n > 0$ for all $n$. We have that for all $n$,
\begin{equation}
    \dot{\eta}(t) \geq g(\eta(t)) > g(\eta(t)) - \epsilon_n \quad \forall t \in [0, \tau_f]
\end{equation}
 and 
 \begin{equation}
     \eta(0) \geq w_0 > w_0 - \epsilon_n.
 \end{equation}
Let $\{r_n\}$ denote a sequence of solutions that satisfy the initial value problem 
\begin{equation}
    \label{eq:minforcompare}
    \dot{r}_n(t) = g(r_n(t)) - \epsilon_n \quad \quad r_n(0) = w_0 - \epsilon_n
\end{equation}
for each $n$. By \cite[Thm 6.2]{bainov2013integral}, the minimal solution for \eqref{eq:minforcompare} exists on $[0, \tau_f]$ for large enough $n$, and to simplify notation, we assume this holds for all $n$. 

We first claim that $\eta(t) > r_n(t)$ for all $t \in [0,\tau_f]$ for any $n$. Suppose for contradiction that $\eta(t) \leq r_n(t)$ for some $t \in [0,\tau_f]$ for some $n$. Then 
\begin{equation}
    T := \inf\{t \in (0,\tau_f] : \eta(t) \leq r_n(t) \}
\end{equation}
is well defined and $T \in (0, \tau_f]$  since $\eta(0) > w_0 - \epsilon_n$. This fact coupled with continuity of $r_n(t)$ and $\eta(t)$ shows that $r_n(T) = \eta(T)$ and $\eta(t) > r_n(t)$ for $t \in [0, T)$. Moreover, $\dot{\eta}(T) > \dot{r}_n(T)$ and since $r_n(T) = \eta(T)$, this implies that $\eta(t) < r_n(t)$ on the interval $(T - \varepsilon, T)$ for some $\varepsilon > 0$. However, this contradicts the definition of $T$, proving the claim. Thus, letting $n \to \infty$, we have $\eta(t) \geq r(t)$ for $t \in [0, \tau_f]$ with $r(t) = \lim_{n \to \infty}{r_n(t)}$. 

Now we claim that $\lim_{n \to \infty}{r_n}$ converges uniformly to some function $r$ on $[0, \tau_f]$ and that $r$ is the minimal solution to \eqref{eq:IVP}. 
Note that because $\{\epsilon_n\}$ is strictly monotone, $\dot{r}_{n+1}(t) > \dot{r}_{n}(t)$ for $t \in [0, \tau_f]$ and $r_{n+1}(0) > r_n(0)$, it holds that $r_{n+1}(t) > r_n(t)$ for $t \in [0, \tau_f]$. By similar reasoning, $r_n(t)$ is also bounded above by any solution of \eqref{eq:IVP}, and in particular by $\Tilde{w}(t)$, so the limit $r(t)$ exists for all $t \in [0, \tau_f]$. Solving for the solution gives 
\begin{equation}
    r_n(t) = w_0 + \int_{0}^{t}{g(r_n(s)) ds} - \epsilon_n(1 + t).
\end{equation}
Hence for $n < m$, 
\begin{align}
\label{eq:cauchy}
    \| r_n(t) - r_m(t) \| \leq &(\epsilon_n - \epsilon_m)(1 + \tau_f) \ + \\  \nonumber &\int_{0}^{\tau_f}{\| g(r_n(s)) - g(r_m(s)) \| ds}
\end{align}
Because $g$ is continuous, it is uniformly continuous on the compact set 
\begin{equation}
    \{z : \min_{t \in [0, \tau_f]}{r_1(t)} \leq z \leq \max_{t \in [0, \tau_f]}{\Tilde{w}(t)} \}.
\end{equation}

Together with the proposed Cauchy criterion \eqref{eq:cauchy}, this implies uniform convergence of $\{r_n\}$ to $r$. Furthermore letting $n \to \infty$ gives 
\begin{equation}
    r(t) =  w_0 + \int_{0}^{t}{g(r_n(s)) ds},
\end{equation}
so $r$ is a solution for \eqref{eq:IVP}. As established previously, $r$ is upper bounded by any solution to \eqref{eq:IVP}, so $r($ is necessarily the minimal solution for \eqref{eq:IVP}, i.e. $r(t) = \Tilde{w}(t)$ and $\eta(t) \geq r(t) = \Tilde{w}(t)$ holds for all $t\in[0, \tau_f]$.

Finally we show the result holds over $[0, \tau)$ with $\tau$ potentially being $\infty$. Suppose for contradiction, $\eta(t) < \Tilde{w}(t)$ for some $t \in [0, \tau)$ and let 
\begin{equation}
    \mathcal{T} := \inf \{t \in [0, \tau) : \eta(t) < \Tilde{w}(t) \}.
\end{equation}
Consider some $\tau_f$ such that $\mathcal{T} < \tau_f < \tau$. Since the minimal solution $\Tilde{w}(t)$ also exists on $[0, \tau_f]$, the preceding argument implies
$\eta(t) \geq \Tilde{w}(t)$ on $[0, \tau_f]$, contradicting the definition of $\mathcal{T}$ and proving $\eta(t) \geq \Tilde{w}(t)$ for all $t \in [0, \tau)$.
\end{proof}
\end{extendedonly}

Motivated by our interest in using scalar differential equations as barrier functions, we are especially interested in scalar systems for which minimal solutions remain nonnegative when initialized at the origin.

\begin{definition}
\label{def:minimal}
A continuous function $\mu: \R \to \R$ is a \emph{minimal function} if the minimal solution  $\Tilde{w}(t)$ defined on $t \in [0,\tau)$ for the initial value problem $\dot{w} = -\mu(w)$, $w(0)=0$ satisfies $\tilde{w}(t)\geq 0$ for all $t\in [0, \tau)$.
\end{definition}

With minimal functions introduced, we can now describe a corresponding barrier function condition.

\begin{definition}
\label{def:MBF}
For the system in \eqref{eq:sys}, a continuously differentiable function $h:\D \to \R$ is a \emph{minimal barrier function} (MBF) if there exists a minimal function $\mu$ that satisfies
\begin{equation}
    \label{eq:thm}
    L_f h(x) \geq -\mu(h(x)) \quad \forall  x \in \D,
\end{equation}
\change{where $L_f h(x) = \frac{\partial h}{\partial x}(x) f(x)$ denotes the Lie derivative.}
\end{definition}

The notion of a minimal barrier function allows us to then establish invariance of $\S$.

\begin{thm}
\label{thmbarrier}
Consider the system \eqref{eq:sys} and a nonempty $\S=\{x \in \D: h(x) \geq 0 \}$ for some continuously differentiable $h: \D \to \R$. If $h$ is a MBF as in Definition \ref{def:MBF}, then $\S$ is positively invariant.
\end{thm}

\begin{proof}
\change{
Let $x(t)$ be a solution defined on $[0, \tau_{\max})$ to \eqref{eq:sys} 
with any $x(0) = x_0 \in \S$. Observe that $h(x(0)) \geq 0$. Consider the comparison system $\dot{w} = -\mu(w)$ with $w(0) = 0$.
We first show that $\Tilde{w}(t)$ is defined on $[0, \tau_{\max})$ as well. Suppose that $\Tilde{w}(t)$ is only defined on $[0, \tau^*)$ for $\tau^* < \tau_{\max}$. Since $\mu$ is a minimal function, $\Tilde{w}(t) \geq 0$ and $\lim_{t\to\tau^*}\Tilde{w}(t)=\infty$ \cite[Corollary 1.1.2]{lakshmikantham1969differential}. By Proposition \ref{prop:compare}, $h(x(t))\geq \Tilde{w}(t)$ for $t \in [0, \tau^*)$, implying $\lim_{t\to \tau^*}h(x(t))=\infty$ and diverges. Since $x(t) \in \D$ for $t\in[0,\tau^*]$, $h(x(t))$ is a well-defined continuous map from $t\in[0,\tau^*]$ to $\R$. Contradiction ensues, as the image of $[0,\tau^*]$ under the continuous map $h$ is compact and therefore bounded, and the claim is shown.
Now, Proposition \ref{prop:compare} gives that $h(x(t)) \geq w(t) \geq 0$, implying that $x(t) \in \S$ for all \change{$t \in [0, \tau_{\max})$}. Therefore $\S$ is positively invariant.}
\end{proof}

The main theoretical component for Theorem \ref{thmbarrier} comes directly from the differential inequality in Proposition \ref{prop:compare}, rather than using the classical argument by Nagumo. 
In this regard, MBFs highlight the strong connection between set invariance and differential inequalities. 
\begin{shortonly}
\change{Furthermore, if the additional assumption that $\mu(w) \leq 0$ ($< 0$) for $w \leq 0$ holds, then the utility of the MBF is similar to a set-based Lyapunov function for guaranteeing stability (asymptotic stability). A more detailed discussion can be found in the Appendix B of the extended version.}
\end{shortonly}

\subsection{Necessary and sufficient conditions for minimal functions}

Relaxing the standard Lyapunov condition $L_f h \geq 0$ for all $x \in \D$ has been studied in a number of works, see e.g. \cite{kong2013exponential}, \cite{ames2017control}. By considering a larger class of comparison functions to lower bound the flow, i.e. $L_f h(x) \geq -\phi(h(x))$, it becomes possible to ensure invariance without requiring stability, allowing for a larger design space when constructing valid barrier functions. 

We remark that minimal functions, by definition, represent the largest possible class of comparison functions, \change{and in a  certain sense, are the most general class of functions that can be used in a barrier function condition.}
However, since minimal functions are implicitly defined through the resulting nonnegative solutions, checking if a function is indeed minimal is not immediately apparent. Thus the next theorem presents verifiable conditions on $\mu$ to ensure that it is a minimal function. It is important to note that since the next theorem is necessary and sufficient, it represents the tightest possible conditions for a minimal function.
The crux of the theorem comes from uniqueness results in \cite{agarwal1993uniqueness}.

\begin{thm}
\label{thm:minimal}
A continuous function $\mu : \R \to \R$ is a minimal function if and only if one of the following cases is satisfied:
\begin{enumerate}
    \item $\mu(0) < 0$ 
    \label{item:1}
    
    \item $\mu(0) = 0$ and there exists  $\varepsilon>0$ such that $\mu(w) \leq 0$ for all $w \in [-\varepsilon, 0)$
    \label{item:2}
    
    \item $\mu(0) = 0$ and for every $\varepsilon>0$, there exists some $w', w''$ in $[-\varepsilon, 0]$ such that $\mu(w') > 0$ and $\mu(w'') < 0$
    \label{item:3}
    
    \item $\mu(0) = 0$ and there exists $k > 0$ such that for all $\varepsilon$ with $0 < \varepsilon < k$, $\mu(w) \geq 0$ on $[-\varepsilon,0]$ and $-1/\mu(w)$ is not integrable on $[-\varepsilon, 0]$, i.e. $\int^{-\varepsilon}_{0}{-\frac{dw}{\mu(w)}}$ is divergent
    \label{item:4}
    
\end{enumerate}
\end{thm}

\begin{extendedonly}
\begin{proof}
Let the system be $\dot{w} = -\mu(w)$ and the corresponding minimal solution be $\Tilde{w}(t)$ for the initial condition $w(0) = 0$.
We use $a.e.$ for abbreviation for almost everywhere.
First we show the sufficient direction. 

\emph{Case 1)} Since $\mu(0) \neq 0$, the minimal solution $\Tilde{w}(t)$ is unique by \cite[Thm 1.2.7]{agarwal1993uniqueness}. Because $\mu(0) < 0$ and $\mu$ is continuous, there exists  $\varepsilon>0$ such that $\mu(w) < 0$ everywhere on $[-\varepsilon,0]$
Suppose for contradiction, there exists $\tau > 0$ where $w(\tau) = b<0$. Integration of the system gives
\begin{equation}
\int^{w(\tau) = b}_{w(0) = 0}{-\frac{dw}{\mu(w)}} = \tau - 0.
\end{equation}
The integral on the left is negative since $b < 0$ and $\mu(w) < 0$ on $[-\varepsilon,0]$,
but $\tau > 0$, giving a contradiction. 

\emph{Case 2)} \cite[Thm 2.2.2]{agarwal1993uniqueness} First we show that any solution $w(t)$ is monotone. Suppose for contradiction, that $w(t)$ is not monotone. Then there exists two times $t_1 \neq t_2 \in [0, \tau]$ such that $w^* = w(t_1) = w(t_2)$ and $\dot{w}(w^*) > 0$ and $\dot{w}(w^*) < 0$. But this is a contradiction, since $\dot{w}(t) = -\mu(w(t))$ is a function of $w(t)$. 

Since any solution $w(t)$ is monotone, and we assume that $\mu(w) \leq 0$ for all $w \in [-\varepsilon, 0]$, $w(t)$ is a non-decreasing function. Therefore, there does not exist a time $\tau$ in which $w(t) < w(0) = 0$, and any solution $w(t) \geq 0$ for all $t \geq 0$, including specifically $\Tilde{w}(t)$.

\emph{Case 3)} \cite[Thm 2.2.2]{agarwal1993uniqueness} Suppose for contradiction, there exists $\tau > 0$, where $w(\tau) = b<0$. By assumption, $\mu(w(t))$ switches sign on $[b,0]$, and thus $w(t)$ is not monotone on $[0, \tau]$. However, solutions must be monotone, giving a contradiction.

\emph{Case 4)} \cite[Thm 2.2.2]{agarwal1993uniqueness} Suppose for contradiction, there exists a $\tau > 0$ where $w(\tau) = b \in  [-\varepsilon,0)$ for some $\varepsilon < k$. Let $\bar{\tau} < \tau$ be the greatest point such that $w(\bar{\tau}) = 0$. Take a monotonic sequence $\bar{\tau} < t^k < \tau$ converging down to $\bar{\tau}$. Integration of the ODE gives 
\begin{equation}
    \lim_{t^k \to \bar{\tau}^+}{\int^{b}_{w(t^k)}{-\frac{dw}{\mu(w)}}} = \lim_{t^k \to \bar{\tau}^+}{\tau - t}.
\end{equation}

Since $w(t)$ is monotone, the set of $w$ in $[b, 0]$ in which $\dot{w} = -\mu(w) = 0$ is a measure zero set.
Thus, it must be that $\mu(w) > 0 \ a.e.$ on $w \in [b, 0]$ and there exists a set $G = [b, 0] \setminus Z$ such that $\mu(w) > 0$ for all $w \in G$, where $Z$ is a measure zero set. Let $G_k = [b, w(t^k)] \cap G$. Then the integral $\int_{G_k}{-\frac{dw}{\mu(w)}}$ converges to the improper integral $\int^{b}_{0}{-\frac{dw}{\mu(w)}}$ via the monotone convergence theorem. By assumption, $-1/\mu(w)$ is not integrable on $[b, 0]$, since $-b < k$ and $\int^{b}_{0}{-\frac{dw}{\mu(w)}}$ diverges. However, $\lim_{t^k \to \bar{\tau}^+}{\tau - t}$ is bounded above by $\tau$ and therefore converges, giving a contradiction.

Now we consider the necessary direction. Assume all conditions do not hold. Then either $\mu(0) > 0$ or $\mu(0) = 0$, $-1/\mu(w)$ is integrable on $[-\varepsilon,0]$, and $\mu(w) > 0$ on $[-\varepsilon,0] \setminus Z$ for some $\varepsilon > 0$ and some measure zero set $Z$. If $\mu(0) > 0$, there exists a $[-\varepsilon,0]$ where $\mu(w) > 0$. Note also the minimal solution $\Tilde{w}(t)$ is unique by \cite[Thm 1.2.7]{agarwal1993uniqueness}. Choose a point $b \in U^-_{\varepsilon}$ and integrate to get $\int^{b}_{0}{-\frac{dw}{\mu(w)}} = \tau$. Because $\mu(w) > 0$ on $[b, 0]$, we can set $\tau$ to the value of the integral. Therefore this equation defines the solution where $w(\tau) = b < 0$. 

Now consider the second condition, which follows from \cite[Thm 1.4.3]{agarwal1993uniqueness}. Let $G^t = [w(t), 0] \setminus Z$. The integral of the ODE is $\int_{G^t}{-\frac{dw}{\mu(w)}} = t$ for $-\varepsilon \leq  w(t) < 0$. Since $-1/\mu(w)$ is integrable on $[-\varepsilon,0]$ by assumption, the integral converges and defines a family of solutions $w_c(t)$ satisfying
\begin{align}
    \begin{cases}
    w_c(t) = 0 \quad &t \leq c \\
    \int_{G^t}{-\frac{dw}{\mu(w)}} = t - c \quad &t > c
    \end{cases}
\end{align}
for $c \in \R^+ \cup \{\infty\}$. Since $\mu(w) > 0$ on $G^t$, taking $c = 0$ gives the minimal solution in which $w_c(t) < 0$ for $t > 0$.
\end{proof}
\end{extendedonly}

Case \ref{item:1} and Case \ref{item:2} are similar in vein to the standard Lyapunov condition, as $L_f h \geq - \mu(h) \geq 0$ on $h \in [-\varepsilon, 0]$ for some $\varepsilon > 0$. Case \ref{item:3} considers the  case when $\mu$ changes sign infinitely often. Case \ref{item:4} relaxes the usual locally Lipschitz condition to a one-sided nonintegrability condition to handle a more general class of comparison functions.

\change{Uniqueness functions have also appeared in the literature as a means for establishing invariance \cite{maghenem2019characterizations, maghenem2019sufficient, redheffer1975flow}. Essentially, $g$ is a uniqueness function if any continuously differentiable $\eta(t)$ satisfying $\eta(0)=0$ and $\dot{\eta}(t)= g(\eta(t))$ for all $t$ must necessarily be the unique solution $\eta(t) \equiv 0$ \cite{ladde1974flow}. It can be seen that all continuous uniqueness functions are minimal functions, but $\mu(w) = -w^{\frac{2}{3}}$ is an example of a minimal function that is not a uniqueness function. 
In this way, the definition of minimal functions captures the essential philosophy of barrier functions: invariance is certified by the nonexistence of solutions to the comparison system in \eqref{eq:IVP} that become strictly negative, and nonunique nonnegative solutions are not relevant to establishing invariance.}

\change{But}, if a minimal function is \change{a priori known} to induce unique solutions, then only the condition that $\mu(0) \leq 0$ needs to be checked. More specifically, it can be verified that all locally Lipschitz minimal functions with $\mu(0) \leq 0$ do indeed satisfy the hypotheses of Theorem \ref{thm:minimal}.

\begin{corollary}
\label{prop:lipmin}
Any locally Lipschitz continuous function $\mu_L: \mathbb{R}\to\mathbb{R}$ with $\mu_L(0) \leq 0$ satisfies the hypotheses of Theorem \ref{thm:minimal} and therefore is a minimal function.
\end{corollary}
\begin{extendedonly}
\begin{proof}
Notice that if, for any $\varepsilon > 0$, there exists a positive measure set $P \subset [-\varepsilon, 0]$ where $\mu_L(w) \leq 0$ for all $w \in P$, then $\mu_L$ has to satisfy one of the cases \ref{item:1}--\ref{item:4} of Theorem \ref{thm:minimal} and is necessarily a minimal function. Thus, assume $\mu_L$ is not a minimal function so that $\mu_L(0) = 0$ and there exists a constant $a > 0$ such that $\mu_L(w) > 0 \ a.e. $ for $w \in [-a, 0]$.
Since $\mu_L$ is locally Lipschitz, there exists a neighborhood $U$ around 0 such that $\mu_L$ is Lipschitz on $U$. Choose $0 < k \leq a$ such that
\begin{equation}
    [-k , 0] \subset [-a, 0] \  \cap \ U.
\end{equation}
Since $\mu_L$ is Lipschitz on $[-k, 0]$,
\begin{equation}
    \| \mu_L(w)\| \leq L \| w \|
\end{equation}
for all $w \leq k$ for some Lipschitz constant $L$. Because $\mu_L > 0 \  a.e.$ on $[-k, 0]$, it follows that $1/\mu_L(w) > 1/(Lw) \ a.e.$ on $[-k, 0]$. Then $1/\mu_L(w)$ is not integrable since $1/(Lw)$ is not integrable on any $[-\varepsilon, 0]$ for $\varepsilon \leq k$ and $\mu_L$ then satisfies Case 4 of Theorem \ref{thm:minimal}. Therefore one of the cases of Theorem \ref{thm:minimal} must hold.
\end{proof}
\end{extendedonly}

\subsection{Examples}
\label{subsec:examp}

The following examples and anti-examples demonstrate the utility of the proposed formulation of minimal barrier functions. We begin with an anti-example that highlights the importance of considering minimal solutions to differential inequalities when constructing comparison systems.

\begin{example}
  \label{ex:counterexample}
  Consider $\dot{x}=f(x)= -1$ for $x\in \mathbb{R}$ and let $h(x)= x^3$. Take
  $\mu(w)= 3(w^{1/3})^2$. Then $L_f h(x)= -\mu(h(x))$ for all $x\in\mathbb{R}$. 
  
  Although the function $\mu$ satisfies $\mu(0) \leq 0$, it is not locally Lipschitz and Corollary \ref{prop:lipmin} does not apply. Moreover, $\mu$ does not satisfy any of the conditions of Theorem \ref{thm:minimal}. Indeed, $\S = \R_{\geq 0}$ is not positively invariant on $\R$.
 
  Further, even though the comparison system $\dot{w}=-\mu(w)$ with the initial condition $w(0)=0$ has a solution $w(t) \equiv 0$, it also has the minimal solution $\Tilde{w}(t)=-t^3$. Considering $x(t)$, the solution to $\dot{x}=f(x)$ with $x(0)=0$, we see that $h(x(t))=\Tilde{w}(t)$, \emph{i.e.}, the barrier function $h$ evaluated along solutions of the system $\dot{x}=f(x)$ just match the minimal solution of the comparison system. 
\end{example}

The following example examines the case where the set $\S$ has corners, but still can be verified using a minimal barrier function.

\begin{example}
\label{ex:nonsmooth}
  Consider the system
  \begin{align}
    \label{eq:sysnonsmooth}
    \dot{x}_1&=-ax_1+bx_2\\
    \dot{x}_2&=cx_1-dx_2
  \end{align}
  where $a,b,c,d\geq 0$. Let a barrier function be $h(x)=x_1x_2$ so that $\S$ is the union of the first and third quadrants of the plane. $L_f h(x) = -ax_1x_2+bx_2^2 + cx_1^2-dx_1x_2 \geq -ax_1x_2-dx_1x_2=(-a-d)h(x)$ so that \eqref{eq:thm} is satisfied with $\mu(w)=(a+d)w$ and $\S$ is positively invariant.
\end{example}

In the next example, it is necessary to consider a non-Lipschitz minimal function to establish forward invariance with a given barrier function. Even though the vector field of the system is Lipschitz, and the barrier function $h$ is smooth, the resulting dynamics for $L_f h$, as a function of $h$, may not be Lipschitz. 

\begin{example}
  \label{ex:bump}
  Consider $\dot{x}=-|x|$ for $x \in \R$ and let 
  \begin{align}
    \label{eq:bumph}
    h(x)=
    \begin{cases}
      \exp(-1/x)&\text{if $x\geq 0$}\\
      -\exp(1/x)&\text{if $x<0$}
    \end{cases}
  \end{align}
  so that $\S = \{x:h(x)\geq 0\}= \R_{\geq 0}$ is indeed invariant. 
\begin{extendedonly}
  Now we calculate
    \begin{align}
    \label{eq:bumphdot}
       L_f h(x) =
    \begin{cases}
      -\exp(-1/x)/x &\text{if $x\geq 0$}\\
      \exp(1/x)/x &\text{if $x<0$}.
    \end{cases}
  \end{align}
The function $h$ is invertible with inverse
  \begin{align}
    \label{eq:43}
    h^{-1}(w)=
    \begin{cases}
      {-1}/{\ln(w)}&\text{if $0\leq w<1$}\\
      {1}/{\ln(-w)}&\text{if $-1<w<0$}.
    \end{cases}
  \end{align}
\end{extendedonly}
  Define a minimal function candidate
  \begin{align}
    \mu(w) = 
    \begin{cases}
      - w \ln(w) &\text{if $0 \leq w <1$}\\
      w \ln(-w) &\text{if $-1< w <0$}
    \end{cases}
\end{align}
and observe that $L_f h(x)=-\mu(h(x))$. We check that $\mu$ is a minimal function. Indeed, $\mu(h)$ is continuous with $\mu(0) = 0$ and $\mu(h) > 0 $ over $\D \setminus \{0\}$. Corresponding to Case \ref{item:4} in Theorem \ref{thm:minimal}, we check that the improper integral
\begin{align}
    \int^{-a}_{0}{-1/\mu(w) dw}
    &= - \ln(\| \ln(-w)\|) |^{-a}_{0}   
\end{align}
diverges to $\infty$ for any $a \in (0, 1)$. Therefore $\mu$ is a valid minimal function. Observe that $\mu$  is not locally Lipschitz at 0. Indeed, it \change{can be established} that there exists no locally Lipschitz minimal function satisfying \eqref{eq:thm} since any such function must be lower bounded by $\mu$ constructed above and be non-positive at the origin. \begin{shortonly}
\change{Explicit calculations are shown in the extended version.}
\end{shortonly}
\end{example}

\subsection{Comparing to Zeroing Barrier Functions}
\label{subsec:zero}

In this section, we compare MBFs to \emph{zeroing barrier functions} (ZBFs) in \cite{ames2019control}, which use extended class $\mathcal{K}$ functions for the class of comparison functions, and is a major inspiration for the work in this paper. \change{A function $\alpha: \R \to \R$ is extended class $\mathcal{K}$ if it is strictly increasing with $\alpha(0) = 0$.} We remark that if $\alpha$ is an extended class $\mathcal{K}$ function, then $-\alpha$ is a minimal function, as guaranteed by Case \ref{item:2} in Theorem \ref{thm:minimal}.

In \cite{ames2019control}, the development of ZBFs requires utilizing Nagumo's Theorem and therefore requires the assumption that $\frac{\partial h}{\partial x}$ does not degenerate to a zero vector on the boundary of the set. By directly invoking a differential inequality, as in Proposition \ref{prop:compare}, we can dispense with this assumption, which is discussed further in Section \ref{sec:Nagumo}.

\begin{shortonly}
However, restricting to extended class $\mathcal{K}$ functions imposes the barrier condition $L_f h(x) \geq -\alpha(h(x)) \geq 0$ \change{for all $x \in \D$, even when}  $h(x) \leq 0$. This requires stability of each level set \change{$\{ x \in \D : h(x) = w \}$ for any $w < 0$}. While this robustness is desirable in some instances, it does not hold for a large class of systems and sets $\S$. For example, considering the simple system $\dot{x} = x$, it can be shown that there does not exist a class $\mathcal{K}$ function $\alpha$ and a ZBF $h(x)$ such that $L_f h(x) \geq -\alpha(h(x))$ holds on all of $\D=\R$ to verify that $\S = \R_{\geq 0}$ is invariant \change{(see Example 5 in the extended version for more details)}.
\end{shortonly}

Moreover, it is not possible to simply restrict the ZBF to be defined only on \change{$ \{ x \in \D : h(x) \geq 0 \} = \S$}. Indeed, this contradicts the hypotheses in Theorem \ref{thmbarrier} that requires \eqref{eq:thm} to hold for all $x\in \D$, and ignoring this requirement can result in the following scenario.

\begin{example}
  \label{ex:counterexample2}
  Consider again Example \ref{ex:counterexample}, and take 
  \begin{align}
    \label{eq:9}
    \alpha(w)=
    \begin{cases}
      3w^{2/3}&\text{ if $w\geq 0$}\\
      -3w^{2/3}&\text{ if $w< 0$}
    \end{cases}
  \end{align}
  so that $L_f h(x)=-\alpha (h(x))$ for all $x\in \S$, although notably the equality does not hold for $x\in \mathbb{R}\backslash \S$. and thus Theorem \ref{thmbarrier} is not applicable since it requires \eqref{eq:thm} to hold for all $x\in \D$. 
  Notice that $\alpha$ is an extended class $\mathcal{K}$ function on $\mathbb{R}$ and that $-\alpha$ is a minimal function.
  While it is tempting to use $\dot{w}=-\alpha(w)$ as a comparison system with $w(0)=h(x(0))$, we obtain the false conclusion that $\S$ is positively invariant.
\end{example}

\begin{extendedonly}
A considerable benefit of using the more general class of minimal functions over extended class $\mathcal{K}$ functions for a comparison system is that using an extended class $\mathcal{K}$ function necessitates that $L_f h > 0$ on $\D \setminus \S$. While  this type of robustness is sometimes desirable, it does not hold in general.

\begin{example}
\label{ex:norobust}
Consider $\dot{x} = x$ for $x\in\mathbb{R}$ and $h(x) = x$, with the corresponding $\S = \{x: x \geq 0 \}$. Therefore, $L_f h(x) = x = h(x)$. In particular, $L_f h(x)<0$ whenever $h(x)<0$, and thus there does not exist an extended class $\mathcal{K}$ function $\alpha$ satisfying $L_f h(x)\geq -\alpha(h(x))$ for all $x\in \mathbb{R}$. However, $\mu(w)=-w$ is a minimal function satisfying $L_f h(x) \geq -\mu(h(x))$ for all $x\in\mathbb{R}$, thus proving invariance of $\S$.
\end{example}
\end{extendedonly}

\section{Discussion on Regularity}
\label{sec:Nagumo}

Arguably, the most common approach for establishing positive invariance of a set $\S$ is to verify, in some appropriate sense, that the velocity field of the system points inwards to $\S$ at each point on the boundary of $\S$. First formalized by Nagumo in \cite{nagumo1942lage} and independently discovered by others, there has since been a volume of work dedicated to making this basic approach precise in various contexts, e.g. \cite{redheffer1975flow}, \cite{blanchini2008set}, \cite{aubin2009viability}. We consider the important specialization of Nagumo's Theorem to the case where $\S=\{x:h(x)\geq 0\}$ for a smooth function $h : \D \to \mathbb{R}$.

For a continuously differentiable function $h:\D\to \mathbb{R}$ for an open set $\D\subseteq \mathbb{R}^n$, $\lambda\in \mathbb{R}$ is a \emph{regular value} of $h$ if $\frac{\partial h}{\partial x} (x)\neq 0$ for all $x\in\{x \in \D :h(x)=\lambda\}$.

We now recall a version of Nagumo's Theorem, vital to the construction of barrier functions in
\cite{prajna2004safety}, \cite{ames2019control}, etc.

\begin{prop}[{\cite[Sec 4.2.1]{blanchini2008set}}]
\label{prop:nagsmooth}
Consider the system \eqref{eq:sys} under the added condition that solutions are unique, and a nonempty set $\S = \{x \in \D : h(x) \geq 0 \}$ for some continuously differentiable $h : \D \to \R$. Further assume that $0$ is a regular value of $h$. Then $\S$ is positively invariant if and only if
  \begin{align}
    \label{eq:nagumo}
    L_f h(x)\geq 0
  \end{align}
  for all $x\in \{x\in\D:h(x)=0\}$.
\end{prop}

Proposition \ref{prop:nagsmooth} provides a powerful result for establishing invariance of $\S$ provided that $0$ is a regular value of $h$.
In this case, we can equivalently state the condition in \eqref{eq:nagumo} as
\begin{equation}
    \label{eq:nagbar}
    L_f h(x) \geq -\phi(h(x)) \quad \forall x \in \D
\end{equation}
where $\phi :\R \to \R$ is any function with $\phi(0) \leq 0$. Notice that $0$ being a regular value allows us to discount much of the structure of $\mu$ defined for Theorem \ref{thm:minimal}. 
\begin{shortonly}
However, as shown in the next theorem, under further mild conditions, the existence of a locally Lipschitz comparison function is guaranteed.
\end{shortonly}
\begin{extendedonly}
However, as shown in Theorem \ref{thm:existloc}, under further mild conditions, the existence of a locally Lipschitz comparison function is guaranteed.
\end{extendedonly}

\begin{extendedonly}
We first provide a construction for a minimal  function given a candidate barrier function, provided one exists.

\begin{lemma}
\label{lem:Gamma}
Given a system of the form \eqref{eq:sys} and a candidate barrier function $h(x)$, let $\Gamma:W \to \R$ be defined as
\begin{equation}
\label{eq:Gamma}
    \Gamma(w) = \inf_{x: h(x) = w}{L_f h(x)}
\end{equation}
where $W = \{h(x):x\in\D\}$ is the range of $h$. If there exists some minimal function that satisfies condition \eqref{eq:thm}, and if $\Gamma$ is continuous, then $-\Gamma$ is also a minimal function that satisfies condition \eqref{eq:thm}.
\end{lemma}
\begin{proof}
Let $\mu$ be a minimal function satisfying \eqref{eq:thm} so that $\mu(h(x)) \geq - L_f h(x) \ \forall x \in \D$. It follows from the definition of $\Gamma$ that
\begin{equation}
    \mu(w) \geq -\Gamma(w)
\end{equation}
for all $w\in W$. Since $\mu$ is a minimal function, $\mu$ satisfies one of the cases in Theorem \ref{thm:minimal}. 

If, for every $\varepsilon > 0$, there exists a positive measure set $P \subset [-\varepsilon, 0]$ where $\mu(w) \leq 0$ for $w \in P$, then $-\Gamma(w) \leq \mu(w) \leq 0$ on $P$ and therefore $-\Gamma$ has to satisfy one of the cases \ref{item:1}--\ref{item:4} of Theorem \ref{thm:minimal} so that $-\Gamma$ is  a minimal function. We can then consider the alternative condition that $-\Gamma(w)$ and $\mu(w) > 0 \ a.e.$ on $[-k, 0]$ for some $k$.
Because $\mu(h(x)) \geq -\Gamma(h(x))$ for all $x\in\D$ and $\Gamma \neq 0 \ a.e.$, then it must be that $1/\mu(h(x)) \leq -1/\Gamma(h(x)) \ a.e$. Because $1/\mu(h(x))$ is nonintegrable and positive a.e., $-1/\Gamma(h(x))$ is nonintegrable as well. So $-\Gamma$ must necessarily satisfy one of the cases if there exists a minimal function that does.
\end{proof}

Therefore, showing that $-\Gamma(w)$ is a minimal function is equivalent to the existence of a minimal function under the assumption of continuity of $\Gamma(w)$. Thus, we will only focus our attention on $\Gamma(w)$. 

To analyze what conditions on the barrier function are necessary for the continuity properties of $\Gamma$, we introduce some tools from topology and optimization. 

First we describe a generalized inverse of $h(x)$ as a point-to-set mapping $h^{-1}: W \rightrightarrows \D$, where $h^{-1}(w)=\{x \in \D: h(x) = w \}$ and $W \subset \R$ is the range of $h$, \emph{i.e.}, $W=\{h(x):x\in \D\}$. We use $W \rightrightarrows \D$ in place of $W \to 2^{\D}$ for ease of notation. The function $\Gamma(w)$ can now be defined as $\Gamma(w) = \inf \{L_f h(x) : x \in h^{-1}(w)\}$.

We now introduce some necessary definitions for $h^{-1}$ to be continuous as a point-to-set map:

\begin{itemize}
\item $h^{-1}$ is \emph{lower semicontinuous} (l.s.c) at $w_0$ if for each open set $G$ s.t. $G \cap h^{-1}(w_0) \neq \varnothing$, there exists a neighborhood $U(w_0)$ s.t. $w \in U(w_0) \implies h^{-1}(w) \cap G \neq \varnothing$ \cite{berge1997topological}. 

\item $h^{-1}$ is \emph{upper semicontinuous} (u.s.c) at $w_0$ if for each open set $G$ s.t. $h^{-1}(w_0) \subset G$, there exists a neighborhood $U(w_0)$ s.t. $w \in U(w_0) \implies h^{-1}(w) \subset G$ \cite{berge1997topological}.

\item $h^{-1}$ is \emph{continuous} at $w_0$ if it is both upper and lower semicontinuous at $w_0$. 

\end{itemize}
If $h^{-1}$ always maps to a single point, then definitions of lower and upper semicontinuity coincide with the standard definitions of continuity for functions \cite{berge1997topological}.

The next theorem gives the necessary conditions for a locally Lipschitz minimal function to exist.

\end{extendedonly}

\begin{thm}
\label{thm:existloc}
Let $h: \D \to \R$ be a twice continuously differentiable function, assume $f$ in \eqref{eq:sys} is locally Lipschitz, and suppose $\Lambda_{\delta} := \{x \in \D : -\delta \leq h(x) \leq \delta \}$ is compact for all $\delta \geq 0$. Further assume that $0$ is a regular value of $h$. Then there exists a locally Lipschitz function $\mu_L$ such that $\mu_L(0) \leq 0$ and 
\begin{equation}
    L_f h(x) \geq -\mu_L(h(x)) \quad \forall x \in \D.
\end{equation}
\end{thm}

\begin{extendedonly}
\begin{proof}
Let $\rho(x, A) = \inf \{ \| x - y \| : y \in A \}$ denote the point-to-set distance from $x\in\D$ to some set $A\subseteq \mathbb{R}^n$. 
Because $0$ is a regular value of $h$, Lyusternik's Theorem \cite{dmitruk1980lyusternik} applies, so that for all $x\in h^{-1}(0)$, there exists a neighborhood $\mathcal{N}_1(x)\subset \mathbb{R}^n$ of $x$, a neighborhood $\mathcal{N}_2(x)\subset \mathbb{R}$ of 0, and a constant $K(x)>0$ such that
\begin{equation}
    \rho(x', h^{-1}(w)) \leq K(x) \| h(x') - w \|
\end{equation}
for all $x'\in \mathcal{N}_1(x)$ and all $w\in \mathcal{N}_2(x)$.

Next, notice that $\bigcup_{x\in h^{-1}(0)}\mathcal{N}_1(x)$ is an open cover of $h^{-1}(0)$, and $h^{-1}(0) = \Lambda_0$ is assumed to be compact. By compactness of $h^{-1}(0)$, there exists a finite subcover, \emph{i.e.}, a finite set of points $\{x_i\}_{i=1}^N\subset h^{-1}(0)$ such that
\begin{equation}
    h^{-1}(0)\subset \bigcup_{i=1}^N \mathcal{N}_1(x_i)=:\mathcal{C}.
\end{equation}

It is shown in Proposition \ref{prop:contminim} that $h^{-1}$ is u.s.c at $0$ under the assumptions of the theorem statement.  By definition of upper semicontinuity, since $\mathcal{C}$ is an open set that covers $h^{-1}(0)$, there exists a neighborhood $V$ of $0$ such that
\begin{equation}
    w \in V \implies h^{-1}(w) \subset \mathcal{C}.
\end{equation}
For some $i\in\{1,\ldots, N\}$, there also exists a neighborhood $X$ of the point $x_i$ and a neighborhood $W$ of $0$ such that the solution for $h(x) = w$ exists for any $w \in W$ and some $x \in X$, due to $0$ being a regular value \cite{aubin2006applied}. Then, for all $w \in W$, $h^{-1}(w)$ is nonempty. Let 
\begin{equation}
    L_1 = \max_{i\in\{1,\ldots,N\}}{K(x_i)}
\end{equation}
and let
\begin{equation}
    U = \bigcap_{i=1}^N{\mathcal{N}_2(x_i)} \cap V \cap W.
\end{equation}
Observe that $U$ is a neighborhood of $0$ since it is a finite intersection of neighborhoods of $0$. Therefore, 
\begin{equation}
    \rho(x', h^{-1}(w)) \leq L_1 \| h(x') - w \|
\end{equation}
for all $x'\in \mathcal{C}$ and all $w\in U$.

Because $f$ and $\frac{\partial h}{\partial x}$ are locally Lipschitz, so is $L_f h$, and thus $L_f h$ is Lipschitz on some compact set $U_c \supset U$ with some Lipschitz constant $L_2$. Now we show that $\Gamma$ defined in \eqref{eq:Gamma} is Lipschitz on $U$. Choose $w_1, w_2\in U$. Now choose $x_1\in h^{-1}(w_1)$ such that $L_f(x_1) = \Gamma(w_1)$. This is possible since $L_f h$ is continuous and $h^{-1}(w_1) \subset \Lambda_{\|w_1\|}$ is compact, so an extremal point exists. Next, choose $x_2\in h^{-1}(w_2)$ such that
\begin{equation}
    \|x_1 - x_2 \| = \rho(x_1, h^{-1}(w_2))
\end{equation}
which is also possible since $h^{-1}(w_2)$ is also compact and $\| x_1 - x_2 \|$ is continuous in $x_2$ for a  fixed $x_1$. Note that
\begin{align}
    \label{eq:ineq1}
    \Gamma(w_2) - \Gamma(w_1) &\leq L_f h(x_2) - L_f h(x_1) \\
    \label{eq:ineq2}
    &\leq \| L_f h(x_2) - L_f h(x_1) \| \\
    \label{eq:ineq3}
    &\leq L_2 \|x_1 - x_2 \| \\
    \label{eq:ineq4}
    &\leq L_2 \rho(x_1, h^{-1}(w_2)) \\
    \label{eq:ineq5}
    &\leq L_1 L_2 \| h(x_1) - w_2 \| \\
    \label{eq:ineq6}
    &\leq L_1 L_2 \| w_1 - w_2 \|.
\end{align}
The inequality \eqref{eq:ineq1} holds from properties of $\inf$ and \eqref{eq:ineq3} is due to $L_f h$ being Lipschitz with a Lipschitz constant $L_2$ on $U_c \supset U$. Note that we previously chose $x_2$ to give the inequality in \eqref{eq:ineq4}. Finally, since $w_1$ and $w_2$ are chosen from $U$, we  apply Lyusternik's theorem for the inequality in \eqref{eq:ineq5}. A similar argument establishes that $\Gamma(w_1) - \Gamma(w_2) \leq  L_1 L_2 \| w_1 - w_2 \|$ i.e, 
\begin{equation}
    \| \Gamma(w_1) - \Gamma(w_2) \| \leq  L_1 L_2 \| w_1 - w_2 \|,
\end{equation}
and thus $\Gamma$ is Lipschitz on $U$ with a Lipschitz constant $L_1 L_2$.

Because $\Lambda_\delta$ is assumed to be compact for all $\delta \geq 0$ and $L_f h$ is continuous, $\Gamma$ is bounded on $[-\delta, \delta]$ for all $\delta \geq 0$ as well. Therefore, there exists a locally Lipschitz function $\mu_L: \R \to \R$ such that $\mu_L(w) \geq -\Gamma(w)$ for all $w\in W$ and, for some neighborhood $U'\subset U$ of $0$, $\mu_L$ restricted to $U'$ is equal to $\Gamma$. Furthermore, $L_f h(x) \geq \Gamma(h(x)) \geq -\mu_L(h(x))$ for all $x\in \D$, so the barrier condition \eqref{eq:thm} is satisfied.
Since $\S$ is assumed to be invariant, $L_f h(x) \geq 0$ for all $x \in h^{-1}(0)$, so $\mu_L(0) = -\Gamma(0) \leq 0$. Therefore $\mu_L$ is locally Lipschitz and $\mu_L(0) \leq 0$.
\end{proof}
\end{extendedonly}

Ensuring smoothness properties of the comparison function is useful in generating constraint-based controllers, which is further discussed in Section \ref{sec:Control}.

Moreover, we have the following immediate corollary.
\begin{corollary}
\label{sufnecbar}
Given the assumptions in Theorem \ref{thm:existloc}, $\S = \{x \in \D :h(x)\geq 0\}$ is positively invariant if and only if $h$ is a minimal barrier function.
\end{corollary}

\begin{extendedonly}
\begin{proof}
Sufficiency comes from Theorem \ref{thmbarrier}, so we only show the necessary direction that $\S$ being invariant implies that $h$ is a minimal barrier function. If the assumptions of Theorem \ref{thm:existloc} are satisfied and if $\S$ is invariant, there exists a $\mu_L$ such that $\mu_L$ is locally Lipschitz with $\mu_L(0) \leq 0$ and $L_f h(x) \geq -\mu_L(h(x))$ for all $x \in \D$. By Corollary \ref{prop:lipmin}, $\mu_L$ is a minimal function and hence, $h$ is a minimal barrier function.
\end{proof}
\end{extendedonly}

\change{Notice that no regularity assumption is made on $h$ in Theorem \ref{thmbarrier}. Removal of this assumption is due to the structural conditions on $\mu$. On the other hand, Corollary \ref{sufnecbar} shows that essentially any candidate barrier function $h$ with $0$ a regular value and $\S$ being compact has a locally Lipschitz comparison function. Therefore we can use Theorem \ref{thmbarrier} without loss of generality from Proposition \ref{prop:nagsmooth} for compact sets,}
and this allows for considering a significant class of sets $\S$ that are not regular. For example, proving invariance of points, cycles, or any other lower dimensional manifold is possible with the theory of minimal barrier functions. In such cases, $0$ cannot be a regular value of the barrier function $h$ because the set $\S=\{x:h(x)\geq 0\}$ has measure zero.
In addition, Examples \ref{ex:nonsmooth} and \ref{ex:bump} provide other cases that can be considered with MBFs for which 0 is not a regular value.

\section{Minimal Control Barrier Functions}
\label{sec:Control}
A major benefit for using differential inequalities defined over the whole domain rather than just a boundary-type condition is that it is more amenable to controlled invariance. In constraint-based control, it is desirable to have constraints on the controller that are applied at every point on the domain rather than just a condition on the boundary of $\S$. If a boundary-type condition is directly applied for controlled invariance, the constraints are only active on a measure zero set, which may introduce discontinuities in the controller and render it sensitive to model and sensor noise.

Extensions of minimal barrier functions to control formulations is direct. In this section, we instead consider a control affine system of the form
\begin{equation}
    \label{eq:contaf}
    \dot{x} = f(x) + g(x)k(x)
\end{equation}
with state $x \in \D$, where $\D \subseteq \R^n$ is assumed to be an open set, a feedback controller $k: \D \to \R^m$, and $f : \D \to \R^n$ and $g: \D \to \R^{n \times m}$ are both assumed to be continuous. We also assume that $k(x) \in U(x)$ for all $x\in\D$, where $U: \D \rightrightarrows \R^m$ is a point-to-set map defining state-based input constraints. Point-to-set maps are denoted with $\rightrightarrows$ for ease of notation. Further define $\mathcal{U}$ as the viable set of continuous controllers
\begin{equation}
\mathcal{U} = \{k \text{ continuous} : k(x) \in U(x) \ \forall x \in \D \}.    
\end{equation}
\begin{shortonly}
In the rest of the section, we assume $U(x)$ defines a set of state-based \change{affine} input constraints of the form
\begin{equation}
    \label{eq:inputinequal}
    U(x) = \{u \in \R^m : A(x)u \preceq b(x) \}
\end{equation}
where $A: \R^n \to \R^{k\times m} $ and $b:\R^n \to \R^k$ are both assumed to be continuous in $x$ and $A(x)_i u \leq b(x)_i$ holds elementwise. \change{General convex input constraints with more detailed proofs are treated in the extended version.}
\end{shortonly}

\begin{extendedonly}
A practical way of representing $U$ is through a set of $q$ inequalities
\begin{equation}
    \label{eq:inputinequal}
    U(x) = \{u \in \R^m : e_i(x, u) \leq 0 \quad i = 1, \ldots, q\}
\end{equation}
where $e_i(x, u) : \D \times \R^m \to \R$ are scalar-valued functions that define state based input constraints. We assume $U$ can be written in this form for the rest of the section.

We further assume $e_i(x, u)$ are \emph{strictly quasiconvex} \cite{evans1970stability} in $u$ for a fixed $x$ and continuous in both $x$ and $u$. 
A strictly quasiconvex function $e: \D \to \R$ satisfies
\begin{equation}
    e(u_1) < e(u_2) \implies e(\theta u_1 + (1-\theta) u_2) < e(u_2)
\end{equation}
for $\theta \in (0, 1)$. We use strictly quasiconvex functions to generalize linear input constraints in the form $A(x)u \preceq b(x)$ to a certain class of convex input constraints. 
\end{extendedonly}

A set $\S \subseteq \D$ is \emph{positively controlled invariant} if there exists a continuous controller $k$ within the possible class of controllers $\mathcal{U}$ such that $\S$ is positively invariant with respect to the closed loop system $\dot{x} = f(x) + g(x)k(x)$ \cite[Def 4.4]{blanchini2008set}.

We now state the corresponding definition of minimal barrier functions for control affine systems.

\begin{definition}
\label{def:MCBF}
For the control affine system in \eqref{eq:contaf}, a  continuously differentiable function $h: \D \to \R$ is a \emph{minimal control barrier function} (MCBF) 
if there exists a minimal function $\mu$ such that for all $x \in \D$,
\begin{equation}
    \label{eq:CBFcond}
    \sup_{u \in U(x)} \left[ L_f h(x) + L_g h(x) u \right] \geq -\mu(h(x)) 
\end{equation}
where $L_f h(x) = \frac{\partial h}{\partial x} f(x)$ and $L_g h(x) = \frac{\partial h}{\partial x} g(x)$ denote corresponding Lie derivatives.
\end{definition}

The set of viable controls is  described by the point-to-set map $K : \D \rightrightarrows \R^m$ given by
\begin{equation}
    \label{eq:feascontset}
    K(x) = \{u \in U(x) : L_f h(x) + L_g h(x) u \geq -\mu(h(x))\}.
\end{equation}

Verification of the existence of controllers with certain properties can be treated as a selection problem, which has been extensively studied in topology \cite{michael1956continuous}. Specifically, the feedback controller $k$ is a \emph{selection} of $K$ if $k(x) \in K(x)$ for all $x \in \D$. Note that for a controller $k$ to render the set $\S$ invariant, it must necessarily be a selection from the point-to-set map $K$.

Additionally, the controller $k$ must come from the set of continuous viable controllers $\mathcal{U}$ in order to guarantee existence of solutions for the closed loop system. Furthermore, continuity of $k$ is also necessary to apply the differential inequality in Proposition \ref{prop:compare} and to satisfy the proposed definition of positive controlled invariance. 

\begin{thm}
\label{thmcntbar}
Given the control affine system \eqref{eq:contaf}, consider a nonempty $\S = \{x \in \D: h(x) \geq 0 \}$ for some continuously differentiable $h : \D \to \R$. If $h$ is a MCBF as in Definition \ref{def:MCBF} and there exists a continuous controller $k \in \mathcal{U}$ such that $k$ is a selection of $K$, then $\S$ is positively controlled invariant.
\end{thm}

\begin{extendedonly}
\begin{proof}
The proof is analogous to the proof of Theorem \ref{thmbarrier}.
\end{proof}
\end{extendedonly}

To guarantee existence of a continuous controller $k$, $K$ being nonempty is not sufficient, and additional conditions on $K$ must be assumed. The next theorem gives sufficient conditions on the existence of a continuous controller $k$ that is a selection of $K$ and therefore can be used to satisfy Theorem \ref{thmcntbar} to render $\S$ positively invariant.

With $U$ characterized as in \eqref{eq:inputinequal}, we also denote the strict interior $K_I: \D \rightrightarrows \R^m$ as
\begin{shortonly}
\begin{align}
    \label{eq:strictintK}
    K_I(x) &= \{u \in \R^m : A(x)u \prec b(x) \nonumber\\ 
    & \qquad \qquad \quad \ L_f h(x) +L_g h(x) u > -\mu(h(x))\}.
\end{align}
\end{shortonly}
\begin{extendedonly}
\begin{align}
    \label{eq:strictintK}
    K_I(x) &= \{u \in \R^m : e_i(x, u) < 0 \quad i = 1, \ldots, q, \nonumber\\ 
    & \qquad \qquad \quad \ L_f h(x) +L_g h(x) u > -\mu(h(x))\}.
\end{align}
\end{extendedonly}
where the input constraints are described with a strict inequality.

\begin{prop}
\label{prop:existcont}
Given $U$ is defined as in \eqref{eq:inputinequal},
if $K_I(x)$ as defined in \eqref{eq:strictintK} is nonempty for each $x$, then there exists a continuous controller $k$ that is a selection of $K$.
\end{prop}
\begin{extendedonly}
\begin{proof}
Define $e_0:\D\times \mathbb{R}^m\to \mathbb{R}$ according to
\begin{equation}
    e_0(x, u) = -L_f h(x) - L_g h(x) u - \mu(h(x)).
\end{equation}
Notice that $e_0$ is also continuous in $x$ and $u$ and strictly quasiconvex in $u$ for each fixed $x$. The viable control map $K$ defined in \eqref{eq:feascontset} can then be described as 
\begin{equation}
    K(x) = \{u \in \R^m : e_i(x, u) \leq 0 \text{ for all } i = 0, \ldots, q\}. 
\end{equation}
Because $K_I(x)$ is assumed to be nonempty for each $x$, and all $e_i$ are continuous and strictly quasiconvex in $u$ for each fixed $x$, the closure $\overline{K_I(x)} = K(x)$ for all $x \in \D$ \cite[Lemma 5]{evans1970stability}, \cite{hogan1973point}, and therefore
$K$ is a l.s.c map \cite[Thm 13]{hogan1973point}. Furthermore, since all $e_i$  are continuous in $x$ and $u$ and strictly quasiconvex in $u$ for each fixed $x$, it holds that $K(x)$ is a closed, convex set in $\R^m$ for all $x \in \D$. Because $K$ is a l.s.c point-to-set map that maps to closed, convex subsets, there exists a continuous controller $k$ that is a selection of $K$ \cite[1.11 Thm 1]{aubin2012differential}.
\end{proof}
\end{extendedonly}

Proposition \ref{prop:existcont} is based on the well known Michael's selection theorem \cite{michael1956continuous}. \change{Proposition \ref{prop:existcont} considers the converse direction of Theorem \ref{thmcntbar}, namely what conditions on $h$ are necessary for there to exist a controller to render $\S$ invariant.}

Usually, a controller $k$ is selected from $K$ based on some optimality criteria. A common approach for safety based control is to first obtain a \emph{nominal controller} $k_{nom}: \D \to \R^m$ that is not verified for either guaranteeing invariance or satisfying input constraints. The nominal controller is then used within a quadratic optimization program (QP) in which $\hat{k}$ is selected from $K$, while minimizing the distance from $k_{nom}(x)$ at each $x$, that is,
\begin{equation}
    \label{eq:uhat}
    \hat{k}(x) = \argmin_{u \in K(x)}{\|u - k_{nom}(x) \|^2}
\end{equation}

Synthesizing controllers in this fashion can allow for real-time control synthesis that satisfies both performance objectives and safety constraints. For a review of applications of this framework, see \cite{ames2019control}. 

Properties of the controller $\hat{k}$ can be analyzed as a selection of $K$, and conditions on $K$ can be formulated to guarantee continuity of $\hat{k}$. In \cite{amesrobust} and \cite{morris2015continuity}, Lipschitz continuity of controllers for quadratic programs regarding a minimum norm controller of $k_{nom} \equiv 0$ for all time was explored. In 
\cite{morris2015continuity}, conditions for pointwise continuity were given, but in this paper, we  show continuity of the controller over the whole domain.

The next theorem gives practical conditions on when the quadratic program in \eqref{eq:uhat} gives a continuous controller. 

\begin{extendedonly}
First we formulate conditions for a controller $\hat{k}$ defined below by a general nonlinear optimization program to be continuous.

\begin{equation}
    \label{eq:uhatloss}
    \hat{k}(x) = \argmin_{u \in K(x)}{\ell(u, k_{nom}(x))}
\end{equation}
where $\ell: \R^m \times \R^m \to \R^+$ is some loss function.

\begin{lemma}
Consider $\hat{k}$ defined by the optimization problem \eqref{eq:uhatloss}. If $K$ defined in \eqref{eq:feascontset} is a continuous, nonempty point-to-set-map, the nominal controller $k_{nom}$ is continuous, $\ell(u_1,u_2)$ is continuous in $u_1$ and $u_2$, and there exists a unique minimizer $u \in K(x)$ of $\ell(\cdot, k_{nom}(x))$ for each $x \in \D$, then the controller $\hat{k}$ is continuous.
\label{lemma:contcont}
\end{lemma}
\end{extendedonly}
\begin{extendedonly}
\begin{proof}
Since $k_{nom}$ is a continuous function, it trivially induces the continuous singleton point-to-set map $k_{nom}'(x) = \{k_{nom}(x)\}$. Since $k_{nom}'$ and $K$ are continuous point-to-set maps, the Cartesian product
\begin{equation}
    K \times k_{nom}': \D \rightrightarrows \R^m \times \R^m
\end{equation}
is a continuous point-to-set map as well \cite[Sec 6.4, Thm 4, 4']{berge1997topological}. Define
\begin{equation}
    V(x) = \inf \{\ell(u_1, u_2): (u_1, u_2) \in K(x) \times k_{nom}'(x)\}
\end{equation}
and let the optimal selection function be 
\begin{equation}
        \Phi(x) = \{
    (u_1, u_2) \in K(x) \times k_{nom}'(x)
    : \ell(u_1, u_2) = V(x) \}.
\end{equation}
By assumption of a unique minimizer $u_1$ for a fixed $u_2=k_{nom}(x)$, $\Phi$ is a singleton point-to-set map. We show that $\Phi$ is a u.s.c map. Since $K \times k_{nom}'$ is a continuous point-to-set map and $\ell$ is a continuous function, $V$ is a continuous function \cite[Max Thm 4.2]{berge1997topological}. The point-to-set map
\begin{align}
    \Delta(x) = \{(u_1, u_2) &\in K(x) \times k_{nom}'(x):  \\ 
    \nonumber &V(x) - \ell(u_1, u_2) \leq 0\}
\end{align}
is a closed point-to-set map, since $V$ is continuous in $x$ and $\ell$ is continuous in $u_1$ and $u_2$ \cite[Thm 10]{hogan1973point}. Notice that
\begin{equation}
    \Phi(x) = \left(K(x) \times k_{nom}'(x) \right)\cap \Delta(x).
\end{equation}
Because $K$ is assumed to be continuous map and $\Delta$ is a closed map, $\Phi$ is an u.s.c map \cite[Sec 6.1 Thm 7]{berge1997topological}. Since $\Phi$ is an u.s.c singleton point-to-set map, it directly induces a continuous function $\Phi'$ defined by $\Phi(x)=\{\Phi'(x)\}$
\cite{berge1997topological}. Take
\begin{equation}
    \hat{k}(x) = (p \circ \Phi')(x)
\end{equation}
where $p$ is the projection function from $(u_1, u_2)$ to $u_1$.
It follows that the optimal controller $\hat{k}$ is a continuous function since $\Phi'$ is continuous. 
\end{proof}

Now we can formulate more specific conditions for a controller defined by a quadratic program as in \eqref{eq:uhat}.
\end{extendedonly}

\begin{thm}
\label{prop:quadprog}
Given $U$ is defined as in \eqref{eq:inputinequal} with $\bigcup_{x \in \D} U(x)$ being compact and $k_{nom}$ is continuous in $x$, if $K_I(x)$ as defined in \eqref{eq:strictintK} is nonempty for each $x$, the controller $\hat{k}$ defined by the  quadratic program in \eqref{eq:uhat}
is continuous.
\end{thm}

\begin{shortonly}
\begin{proof}
\change{
We first claim that $K$ is upper (u.s.c) and lower semicontinuous (l.s.c) and thus a continuous point-to-set map (see section \ref{sec:Nagumo} in the extended version for definitions). Let $K(x)$ be written as the inequality constraint $G(x, u) \leq 0$ with $G(x, u) = [A(x); -L_g h(x)] u - [b(x) ; -\mu(h(x)) - L_f h(x)]$. As $G$ is continuous in both $x$ and $u$, $K$ is a closed map \cite[Thm 10]{hogan1973point} and since $K$ is assumed to map into the compact set $\bigcup_{x \in \D} U(x)$, $K$ is u.s.c \cite[Sec 6.1 Thm 7]{berge1997topological}. Furthermore, $K_I(x)$ is assumed to be nonempty for each $x$ and $K$ is affine in $u$ for each fixed $x$. Therefore, the closure of $K_I(x) = K(x)$ \cite{hogan1973point} and then $K$ is l.s.c \cite[Thm 13]{hogan1973point} as well, and the claim is shown.}

\change{The induced point-to-set map $k'_{nom}(x) = \{ k_{nom}(x)\}$ is continuous by assumption of continuity of $k_{nom}$. Thus, the product map $K \times k'_{nom}$ is continuous \cite[Sec 6.4 Thm 4]{berge1997topological}. Since $K(x) \times k'_{nom}(x) \subset \R^m \times \R^m$ is convex and compact and $\vnorm{u - k_{nom}(x)}^2$ is strictly convex in $u$ for a fixed $x$, there is a unique minimizer of $K \times k'_{nom}$ for each $x$. So $\argmin_{u \in K(x)}{\vnorm{u - k_{nom}(x)}^2} \times k_{nom}(x)$ is a continuous function \cite[Sec 6.3 Max Thm ]{berge1997topological}. Projecting down to the first argument preserves continuity, and therefore, by definition, $\hat{k}$ is continuous. }
\end{proof}
\end{shortonly}

\begin{extendedonly}
\begin{proof}
It is shown in the proof of Proposition \ref{prop:existcont} that $K$ is an l.s.c point-to-set map that maps to convex sets in $\R^m$. As all $e_i$ are continuous, $K$ is a closed mapping \cite[Thm 10]{hogan1973point}, and since $K$ is assumed to map into the compact set $\bigcup_{x \in \D} U(x)$, $K$ is u.s.c \cite[Sec 6.1 Thm 7]{berge1997topological}. Therefore $K$ is a continuous point-to-set map. Since $K$ maps to a convex set and $\ell(u, k_{nom}(x)) = \begin{Vmatrix} u - k_{nom}(x) \end{Vmatrix}^2$ is strictly convex in $u$ for a fixed $k_{nom}(x)$ for each $x \in \D$, there exists at most one solution to the quadratic program. Because $K$ is a continuous, nonempty point-to-set map, $k_{nom}$ is continuous in $x$, and $\ell$ is continuous in $u_1$ and $u_2$ and there is a unique minimizer of $\ell(\cdot, k_{nom}(x))$ for each $x \in \D$, by Lemma \ref{lemma:contcont}, $\hat{k}$ is a continuous controller that is a selection of $K$.
\end{proof}
\end{extendedonly}

If our viable control set $K(x)$ is compact and strictly feasible everywhere, Theorem \ref{prop:quadprog} ensures that the optimal controller $\hat{k}$ is continuous. Since $\hat{k}(x) \in K(x)$ for all $x \in \D$ as well, $\hat{k}$ is applicable to Theorem \ref{thmcntbar}, guaranteeing positive invariance of $\S$.

\begin{extendedonly}
The next example verifies that a program with a quadratic cost and linear constraints returns a continuous controller.
\begin{example}
Consider a system $\dot{x} = u$ for the domain $\D = (-1, \infty)$, a barrier function $h(x) = x$, and a minimal function $\mu(w) = w$. Furthermore, let the input constraints be a state-independent box constraint $U(x) =  \{u \in \R : -1 \leq u \leq 1 \}$ and the desired nominal controller be $k_{nom}(x) \equiv 0$. Then we can define the following quadratic program to generate an optimal controller 
\begin{align*}
    \hat{k}(x) = \argmin_{u \in \R^m} & \  \begin{Vmatrix} u \end{Vmatrix}^2 \quad \\\
   \mathrm{s.t.} \quad   & u \geq \frac{-L_f h(x) - \mu(h(x))}{L_g h(x)} = -x \\
   & u \geq -1 \\
   -& u \geq -1.
\end{align*}
Solving for the explicit controller gives 
\begin{align}
    \hat{k}(x) = 
    \begin{cases}
    -x & 0 \geq x > -1\\ 
    0 & x \geq 0
    \end{cases}
\end{align}
Notice all of the assumptions for Proposition \ref{prop:quadprog} are satisfied, and indeed the resulting controller $k$ is continuous in $x$. We also observe that there is no feasible solution on the interval $\{x < -1 \}$.
\end{example}
\end{extendedonly}

\section{Conclusion}
\label{sec:Conclusion}
This paper presents minimal barrier functions, which stem from scalar differential inequalities, to give the minimum assumptions for utilizing a continuously differentiable barrier function. We have characterized a class of comparison systems viable for verifying invariance of sets defined via a barrier function inequality and have proposed equivalent computable conditions. 
By formulating necessary and sufficient conditions for minimal barrier functions, the relation to Nagumo's theorem is also elucidated.
We then directly extend minimal barrier functions to control formulations and propose relevant conditions for the existence of valid continuous controllers.
By characterizing this relationship with the classical approach of verifying set invariance and rooting the proposed formulation directly in differential inequalities, this paper aims to provide a theoretical foundation for minimal barrier functions. \change{Possible extensions include generalizing minimal barrier functions to hybrid systems.}

\bibliographystyle{ieeetr}
\bibliography{references.bib}

\begin{extendedonly}
\appendices

\section{Extending to Time-Varying Barrier Functions}

In this appendix, we extend the above results  to nonautonomous systems and/or nonautonomous barrier functions. This is crucial when either the vector field of the system or the set under inspection is a function of time. In this section, we study the time varying system
\begin{equation}
    \dot{x} = f(t, x)
    \label{eq:tvsys}
\end{equation}
where $f: [0, \infty) \times \D \to \R^n$ is continuous in $t$ and in $x$. A solution $x(t)$ is defined on a maximum time interval $I[x(\cdot)]=[t_0,\tau_{\max})$ such that $x(t) \in \D$ and \eqref{eq:tvsys} is satisfied for $t \in I[x(\cdot)]$ with an initial condition $x(t_0) = x_0\in \D$ for $t_0\geq 0$, and the solution cannot be extended for time beyond $\tau_{\max}$.

Consider $\S \subseteq [0, \infty) \times \D$ so that $\S_t := \{x \in \D: (t, x) \in \S \}$ is nonempty for all $t\geq 0$. $\S$ is \emph{positively invariant} for \eqref{eq:tvsys} if for all $t_0\geq 0$ the condition $x(t_0)\in \S_{t_0}$ implies all corresponding solutions $x(t)$ satisfy $x(t)\in \S_t$ for all $t\in [t_0,\tau_{\max})$ \cite{blanchini2008set}. $\S$ is \emph{weakly positively invariant} for \eqref{eq:tvsys} if for all $t_0\geq 0$ the condition $x(t_0)\in \S_{t_0}$ implies the existence of a solution $x(t)$ that satisfies $x(t)\in \S_t$ for all $t\in [t_0,\tau_{\max})$. We assume that $\S = \{(t,x) \in [0, \infty) \times \D : h(t, x) \geq 0 \}$  for a function $h$ that is continuously differentiable in both arguments. 

Given a scalar initial value problem $\dot{w}=g(t,w)$, $w(t_0)=w_0$ with $g:[0,\infty)\times W\to \mathbb{R}$ continuous in $t$ and $w$ and for open set $W\subseteq \mathbb{R}$, $t_0\geq 0$, and $w_0\in W$, solutions and minimal solutions are defined analogously to the definitions in Section \ref{sec:Background}.

We slightly modify the definition of minimal barrier functions for time varying formulations.
\begin{definition}
\label{def:minimalvarying}
A continuous function $\mu: [0, \infty) \times \R \to \R$ is a \emph{time varying minimal function} if any minimal solution $\Tilde{w}(t)$ defined on $t \in [t_0, \tau)$ to the initial value problem $\dot{w} = -\mu(t, w) $, $w(t_0) = 0$  satisfies $w(t) \geq 0$ for all $t \in [t_0, \tau)$ for any $t_0 \in [0, \tau)$.
\end{definition}

A sufficient condition for time varying minimal functions is that solutions to $\dot{w}=-\mu(t,w)$ are unique for any initial conditions and that $\mu(t, 0) \leq 0 \ \forall t \geq 0$. Finding necessary and sufficient conditions analogous to Theorem \ref{thm:minimal} is challenging since separation of variables is not possible to get an explicit integral expression.

The following definition parallels Definition \ref{def:MBF} for the time varying case.

\begin{definition}
\label{def:TMBF}
For the system in \eqref{eq:tvsys}, a continuously differentiable function $h:[0, \infty) \times \D \to \R$ is a \emph{time varying minimal barrier function} (TMBF) if there exists a time varying minimal function $\mu$ that satisfies
\begin{align}
    \label{eq:thmvarying}
    \frac{\partial h}{\partial t}(t,x) + L_f h(t, x)  \geq -\mu(t, h(t, x)) \quad \forall (t, x) \in [0, \infty) \times \D.
\end{align}
\end{definition}

\begin{thm}
\label{thmbarriervarying}
Consider the system \eqref{eq:tvsys} and a nonempty $\S = \{(t, x) \in [0, \infty) \times \D : h(t, x) \geq 0 \}$ for some continuously differentiable $h: [0, \infty) \times \D \to \R$. If $h$ is a TMBF as in Definition \ref{def:TMBF}, then $\S$ is positively invariant.
\end{thm}
\begin{proof}
Analogous to the proof of Theorem \ref{thmbarrier}.
\end{proof}

\begin{rem}
Time varying systems and/or barrier functions can also be handled by transforming to a time invariant system by appending an indicator state $\dot{\theta} = 1$ for time. Then it is possible to apply the standard formulations of minimal barrier function. However, for certain specifications, it is not possible to find a time invariant minimal function and a time varying minimal function is necessary.
\end{rem}

The following examples illustrate scenarios in which a time varying formulation is necessary. We first look at a system with a time varying vector field and a time invariant barrier function.

\begin{example}
\label{ex:tvfield}
Consider $\dot{x} = f(t,x)=xt$ for $x\in\mathbb{R}$ and let $h(x) = x$ so that $\S = \{x : x \geq 0 \}$. Along solutions, $L_f h(x) = xt = h(x)t$. Take $\mu(t, w) = -wt$ so that $L_f h(x) = -\mu(t,h(x))$. The function $\mu$ is smooth in both $t$ and $w$, so $\dot{w} = -\mu(t, w)$ has unique solutions. Since $\mu(t, 0) = 0$, $\mu$ is a time varying minimal function, and $h$ is a time varying minimal barrier function by Theorem \ref{thmbarriervarying} so that $\S$ is positively invariant. 

It can be seen that a time invariant minimal function does not exist for this system and barrier. For $\mu(w)$ to be a minimal function, $\frac{\partial h}{\partial x} (x) f(t,x) = x t \geq -\mu(h(x))=-\mu(x)$ for all $t \geq 0$ and all $x\in\mathbb{R}$, but there does not exist a smooth function $\mu$ satisfying $wt\geq -\mu(w)$ for all $t\geq 0$. Specifically,  the inequality does not hold for any $w\in\R^{-}$.
\end{example}

Now we consider a case where the barrier function is time varying but the vector field is time invariant. These formulations are important, for example, when considering barrier functions as an approach to verify reachability \cite{lindemann2019control}.

\begin{example}
\label{ex:tvbarrier}
Consider $\dot{x} = f(x)= x^3 + x$ for $x\in\mathbb{R}$ and let $h(t, x) = e^{-t}x - e^{-2t}$. At $t = 0$, $\S_0 = \{x : x \geq 1\}$, and $\S_t$  increases to $\{x : x \geq 0 \}$ as $t \to \infty$. Along solutions, $L_f h(t, x) = e^{-t}x^3 + 2 e^{-2t} = e^{-t}(h(x) e^t + e^{-t})^3 + 2 e^{-2t} = h^3 e^{2t} + 3h^2 + 3h e^{-2t} + e^{-4t} + 2 e^{-2t}$. Let $\mu(t, h) = -L_f h(t, x)$ where $L_f h(t, x)$ is understood to be a function of $t$ and $h$ as just computed. Then $\mu(t, h)$ can be verified to be a valid minimal function. However, like Example \ref{ex:tvfield}, no time-invariant minimal function exists for this system and barrier function.
\end{example}

\section{Converting to a Lyapunov function}
\label{app:lyap}
We now relate minimal barrier functions to (set based) Lyapunov functions. Notice that the flow constraint in condition \eqref{eq:thm} is similar to the standard Lyapunov flow constraint. And indeed, as $h$ is a scalar function defining a positively invariant set $\S$, we can also utilize $h$ to verify stability properties of $\S$ as well. 

First we recall some notions of stability for dealing with general sets in $A \subset \R^n$ rather than just points.
Let $\rho(x, A) = \inf \{ \|y - x \| : y \in A \}$ denote the distance from $x$ to a set $A$. 

A closed, invariant set $A \subset \D$ is \emph{(uniformly) stable} if for any $\delta > 0$, there exists $\epsilon > 0$ such that for any $x_0$ that satisfies $\rho(x_0, A) < \epsilon$, the solution $x(t)$ satisfies $\rho(x(t), A) < \delta$ for $t \in I[x(\cdot)]$ \cite[Def 1.6.1]{bhatia2006dynamical}.

A set $A \subset \D$ is \emph{asymptotically stable} if it is uniformly stable and there exists a $\delta > 0$ such that for all $x_0$ with $\rho(x_0, A) < \delta$, $\rho(x(t), A) \to 0$ as $t \to \infty$ \cite[Def 1.6.26] {bhatia2006dynamical}.

The next theorem is similar to Proposition 1 in \cite{xu2015robustness}, but is specific to minimal barrier functions and also considers uniform stability of $\S$.

A continuous function $\alpha : \R^+ \to \R^+$ is of class $\mathcal{K}$ if it is strictly increasing and $\alpha(0) = 0$.

\begin{prop}
\label{prop:robust}
Let $h$ be a minimal barrier function for \eqref{eq:sys} and $\S=\{x \in \D :h(x) \geq 0\}$. Assume there exists class $\mathcal{K}$ functions $\alpha, \beta$ and a constant $\delta > 0$ such that
\begin{equation}
    \label{eq:boundedclassk}
    -\beta(\rho(x, \S)) \leq h(x) \leq -\alpha(\rho(x, \S)) \ \forall x \in \S_{\delta} \setminus \S
\end{equation}
where $\S_{\delta} = \{x \in \D : h(x) \geq -\delta \}$.
If there exists a minimal function $\mu$ for $h$ such that \eqref{eq:thm} is satisfied and $\mu(w) \leq 0$ for all $-\delta \leq w < 0$, then $\S$ is uniformly stable. Furthermore if the system \eqref{eq:sys} is forward complete and $\mu(w) < 0$ for all $-\delta \leq w < 0$, then $\S$ is asymptotically stable.
\end{prop}
\begin{proof}
Define a Lyapunov function candidate
\begin{equation}
    \label{eq:lyapbarrier}
    \mathcal{V}_h(x)=
    \begin{cases}
      0&\text{if $x \in \S$}\\
      -h(x)&\text{if $x \in \D \setminus \S$}.
    \end{cases}
\end{equation}
Notice that for all $x \in \S_\delta$,
\begin{equation}
    \label{eq:classkbound}
    \alpha(\rho(x, \S)) \leq \mathcal{V}_h(x) \leq \beta(\rho(x, \S)).
\end{equation}
Additionally, since $\mathcal{V}_h(x)$ is continuous, $\mathcal{V}_h(x)$ is upper bounded by $\beta(\rho(x, \S))$, and for $\delta > 0$, there exists $\gamma > 0$ such that the set 
\begin{equation}
\label{eq:Sgamma}
\S_{\gamma} = \{x \in \D : \rho(x, \S) < \gamma \}
\end{equation}
is a subset of $\S_\delta$.

We first prove that $\mathcal{V}_h(x)$ satisfies the hypotheses of \cite[Corollary 1.7.5]{bhatia2006dynamical}, namely, that $\mathcal{V}_h(x)$ is upper and lower bounded by class $\mathcal{K}$ functions as in \eqref{eq:classkbound} over $\S_\gamma$ and that $\mathcal{V}_h(x(t)) \leq \mathcal{V}_h(x_0)$ for any $x_0 \in \S_\gamma$ and all $t \in I[x(\cdot)]$, which implies that $\S$ is uniformly stable. If $x_0 \in \S$, then $\mathcal{V}_h(x(t)) = \mathcal{V}_h(x_0) = 0$ for all $t \in I[x(\cdot)]$, as $x(t) \in \S$ for all $t \in I[x(\cdot)]$, which is verified by $h(x)$ being a barrier function for $\S$ with the minimal function $\mu(w)$. If $x(t) \in \S_{\gamma} \setminus \S$ for all $t \in I[x(\cdot)]$, then $\mathcal{V}_h(x(t)) \leq \mathcal{V}_h(x_0)$ for all $t \in I[x(\cdot)]$ because
\begin{equation}
    \dot{\mathcal{V}}_h \leq \mu(-\mathcal{V}_h) \leq 0,
\end{equation}
where the first inequality holds since $\mathcal{V}_h = -h$ on $\S_{\gamma} \setminus \S$ and the second follows by hypothesis. Finally, if $x(\tau)\in \S$ at some time $\tau>0$, then $\mathcal{V}_h(x(t)) = 0$ for all $t > \tau$, so $\mathcal{V}_h(x(t)) \leq  \mathcal{V}_h(x_0)$ for all $t \in I[x(\cdot)]$ as well. So for any $x_0$ in $\S_{\gamma}$, $\mathcal{V}_h(x(t)) \leq \mathcal{V}_h(x_0)$ for all $t \in I[x(\cdot)]$. Therefore $\mathcal{V}_h$ satisfies all the required hypotheses and $\S$ is uniformly stable. 
 
Second, under the further assumptions of the proposition regarding asymptotic stability, we claim that $\mathcal{V}_h$ satisfies all of the hypotheses of \cite[Theorem 1.7.8] {bhatia2006dynamical}, proving that $\S$ is indeed asymptotically stable. In particular, these hypotheses are that there exists an open invariant set $\mathcal{B}$ with $\S_\gamma \subset \mathcal{B}$, where $\S_\gamma$ is as defined in \eqref{eq:Sgamma} and $\gamma > 0$, such that: H1)  $\lim_{t \to \infty}{\mathcal{V}_h(x(t))} = 0$ for any $x_0 \in \mathcal{B}$; H2) $\mathcal{V}_h(x(t)) < \mathcal{V}_h(x_0)$ for all $t \geq 0$ and any $x_0 \in \mathcal{B} \setminus \S$; H3) $\mathcal{V}_h(x)$ is upper and lower bounded by class $\mathcal{K}$ functions as in \eqref{eq:classkbound} over $\mathcal{B}$.

We first claim $\S_{\delta}^{\circ} = \{ x \in \D : h(x) > -\delta \}$ is an open invariant set and observe that there exists a $\gamma > 0$ where $\S_\gamma \subset \S_{\delta}^{\circ}$. By assumption, we have $L_f h(x) \geq -\mu(h(x)) > 0$ for $x\in \D$ such that $-\delta \leq h(x) < 0$. Invoking Case \ref{item:1} in Theorem \ref{thm:minimal} below, $\S_{\delta}$ is invariant, and in fact, $\S_{\delta}^{\circ}$ is invariant as well, due to the fact that $L_f h(x) > 0$ on $h(x) \in [-\delta, 0)$, showing the claim.

If $x_0 \in \S_{\delta}^{\circ} \setminus \S$, then  $\mathcal{V}_h(x(t)) < \mathcal{V}_h(x_0)$ for all $t>0$, since
\begin{equation}
    \dot{\mathcal{V}_h} \leq \mu(-\mathcal{V}_h) < 0,
\end{equation}
where the first inequality holds since $\mathcal{V}_h = -h$ on $\S_{\delta}^{\circ} \setminus \S$ and the second follows by hypothesis, and thus H2) holds. Moreover, we claim that $\lim_{t \to \infty}{\mathcal{V}_h(x(t))} = 0$ for all $x \in \S_{\delta}^{\circ}$. If $x_0 \in \S$, then $\mathcal{V}_h(x(t)) = 0$ for all $t\geq 0$ because $\S$ is invariant. If $x_0 \in \S_{\delta}^{\circ} \setminus \S$, using Proposition \ref{prop:compare} with the comparison system $\dot{w} = -\mu(w)$ and the initial condition $w_0 = h(x_0)$, $h(t) \geq \Tilde{w}(t)$, for the minimal solution $\Tilde{w}(t)$. In addition, as the system \eqref{eq:sys} is assumed to be forward complete, both solutions $\Tilde{w}(t)$ and $h(x(t))$ are defined for all $t \geq 0$. Since $w_0 < 0$ and $\mu(w) < 0$ on $[w_0, 0)$, $\lim_{t \to \infty} \Tilde{w}(t)$ must be non-negative. Thus $\lim_{t \to \infty}{h(x(t))} \geq 0$ and $\lim_{t \to \infty}{\mathcal{V}_h(x(t))} = 0$ as well and H1) holds. Hypothesis H3) holds by the assumption \eqref{eq:classkbound}, and therefore $\mathcal{V}_h$ satisfies all the necessary hypotheses and $\S$ is asymptotically stable.
\end{proof}

\begin{rem}
Notice that if $\S_{\delta} = \{x \in \D : 0 \geq h(x) \geq -\delta \}$ is compact for some $\delta \geq 0$, the condition in \eqref{eq:boundedclassk} holds. This is due to the fact that $\gamma(w) = \inf_{\{x \in \D : 0 \geq h(x) \geq -w\}}{\rho(x, \S)}$ is well defined for $w \leq \delta$, so there exists a class $\mathcal{K}$ function that lower bounds $\gamma(w)$, and similarly for the upper bound. 
\end{rem}

\begin{rem}
The explicit requirement of forward completeness of solutions is only needed if $\S$ is not compact. If $\S$ is compact, this necessarily means that there exists a $\delta > 0$ such that $\S_{\delta}^{\circ}$ is compact as well, since its assumed that $h(x) < -\alpha(\rho(x, \S))$ outside of $\S$. Because it can be shown that $\S_{\delta}^{\circ}$ is invariant, the system \eqref{eq:sys} is automatically forward complete in the domain $\S_{\delta}^{\circ}$.
\end{rem}

Intuitively, $L_f h(x) \geq 0$ on $h \in [-\delta, 0]$ for $\delta \geq 0$ implies stability of the set $\S$, and $L_f h(x) > 0$ implies asymptotic stability. In the special case where $\S$ is a point, we recover the classic Lyapunov condition from Proposition \ref{prop:robust}. Therefore, minimal barrier functions can be regarded as a direct extension of Lyapunov functions.

\section{Comparison Between Boundary Conditions}
\label{sec:Nagumocomp}
In this appendix, we recall a few versions of 
results from Nagumo's theorem and we then compare minimal barrier functions with these alternative approaches.

The following result by Nagumo (and independently discovered by Brezis) gives a tangent condition on the flow of the system relative to the set $\S$ that is necessary and sufficient for positive invariance.

\begin{thm}[Nagumo's Theorem \cite{brezis1970characterization}]
\label{prop:Bouligand}
Given the system \eqref{eq:sys} under the further condition that $f(x)$ is locally Lipschitz, consider a nonempty $\S\subseteq \D$ assumed to be closed relative to $\D$. Then $\S$ is positively invariant if and only if
\begin{equation}
\label{eq:invariant}
\lim_{\epsilon \downarrow 0}{\frac{\rho(x + \epsilon f(x), \mathcal{\S})}{\epsilon}} = 0 \; \text{for all } x\in \S
\end{equation}
where $\rho(x, \S) = \inf\{\vnorm{y - x} : y \in \S\}$ for $x\in \mathcal{D}$.
\end{thm}

Another condition which uses normal vectors instead of tangent spaces is given below. A vector $n$ is an \emph{outer normal} to $\S$ at $x$ if $n \neq 0$ and if the closed ball with the center $x + n$ and radius $\| n \|$ has exactly one point in common with $\S$ which is $x$.

\begin{prop}[\cite{bony1969principe}]
\label{prop:normal}
Given the system \eqref{eq:sys} under the further condition that $f$ is locally Lipschitz, consider a nonempty $\S\subseteq \D$ assumed to be closed relative to $\D$. Then $\S$ is positively invariant if and only if
\begin{equation}
\label{eq:invariant2}
n\cdot f(x) \leq 0 \quad \forall x \in \S, \  \forall n \in N(x)
\end{equation}
where $N(x)$ is the set of outer normal vectors to $\S$ at $x$. 
\end{prop}

In Theorem \ref{prop:Bouligand} and Proposition \ref{prop:normal}, it is assumed that the vector fields are locally Lipschitz.
Interestingly, uniqueness functions are also used in  \cite{redheffer1975flow}, \cite{redheffer1972theorems} to relax locally Lipschitz assumptions for the flow of the system for unique solutions, and in \cite{maghenem2019sufficient} to relax an inequality similar to \eqref{eq:invariant2}.

Extensions of invariance results from the standard Nagumo's theorem usually relax three conditions, specifically a smoothness assumption of the boundary of $\S$, a tangent condition on the flow of the system at the boundary, and a uniqueness assumption on the dynamics of the system \eqref{eq:sys} \cite{redheffer1975flow}. By adding a smoothness condition for differentiability of the minimal barrier function, a comparison argument can be utilized to relax some of the uniqueness assumptions for the resulting comparison system.

\section{Continuous Comparison Functions}
\label{sec:Necessary}

In this appendix, conditions for a continuous comparison function are outlined. In other words, we are interested in showing when there exists a $\phi$ such that 
\begin{equation}
    L_f h(x) \geq -\phi(h(x))
\end{equation}
and $\phi(0) \leq 0$ and $\phi$ is continuous. This is important when defining a constraint based controller, where it is desired that the resulting controller defined by an optimization program is continuous. See section \ref{sec:Control} for further details. 

\begin{prop}
\label{prop:contminim}
Let $h: \D \to \R$ be a continuously differentiable function. If $0$ is a regular value of $h$ and $\Lambda := \{x \in \D : -\delta \leq h(x) \leq \delta \}$ is compact for all $\delta \geq 0$, then there exists a continuous function $\phi$ such that $\phi(0) \leq 0$ and
\begin{equation}
    L_f h(x) \geq -\phi(h(x)) \quad \forall x \in \D
\end{equation}
\end{prop}
\begin{proof}
We first claim that $h^{-1}$ is a continuous point-to-set map on some open set $U$ containing $0$. Let $W=\{h(x):x\in\D\}$ and $\S=\{x\in \D: h(x)\geq 0\}$. Define two point-to-set maps $h^+,h^-:W\rightrightarrows \D$ according to
\begin{align}
    h^{+}(w) &= \{x \in \D : -h(x) + w \leq 0 \},\\
    h^{-}(w) &= \{x \in \D : h(x) - w \leq 0 \},
\end{align}
and observe that $h^{-1}(w) = h^{+}(w) \cap h^{-}(w)$ for all $w\in W$. Because $0$ is assumed to be a regular value of $h$ and $\partial \S$ is nonempty, $h^+(w)$ and $h^-(w)$ are both nonempty for all $w$ in some neighborhood $V$ of $0$. The functions $-h(x) + w$ and $h(x) - w$ are both continuous on $W \times \D$, so $h^+$ and $h^-$ are closed on $V$ \cite[Thm 10]{hogan1973point}, and so is $h^{-1}$ due to \cite[6.1 Thm 5]{berge1997topological}. Since $h^{-1}(w)$ is assumed to map into the compact set $\Lambda$ for $w \in [-\delta, \delta]$, $h^{-1}(w)$ is also u.s.c  on $V \cap [-\delta, \delta]$ \cite[6.1 Corollary to Thm 7]{berge1997topological}. 

Because $0$ is a regular value of $h$ and $h$ is continuously differentiable, $\frac{\partial h}{\partial x}$ is constant rank in a neighborhood $\mathcal{N}(x)$ of each $x \in h^{-1}(0)$. Defining $G = \bigcup_{x \in h^{-1}(0)}{\mathcal{N}(x)}$ gives an open cover of $h^{-1}(0) \subset G$ and since $h^{-1}$ was shown to be u.s.c at $0$, there exists a neighborhood $V'$ of  $0$ where $h^{-1}(w) \subset G$ for $w \in V'$ by definition of u.s.c. Therefore for any $w \in V'$, $w$ is a regular value of $h$. Now  define 
\begin{align}
    h^+_I(w) &= \{x \in \D : -h(x) + w < 0 \},\\
    h^-_I(w) &= \{x \in \D : h(x) - w < 0 \}.
\end{align}
We now show $\overline{h^+_I(w)} \supseteq h^+(w)$ for all $w \in V'$. Let $w\in V'$ and let $x^*\in h^+(w)$. If $h(x^*) > w$, then by definition, $x^* \in h^+_I(w)$. If $h(x^*) = w$ and $w \in V'$, then $w$ is a regular value of $h$ and there exists a direction $d\in \mathbb{R}^n$ satisfying $\frac{\partial h}{\partial x}(x^*) d > 0$. Let 
\begin{equation}
    F(a) = h(x^* + ad)
\end{equation}
and notice that $F$ is continuous, so $\lim_{a \to 0^+}{F(a)} = w$ and for $a > 0$ sufficiently close to $0$, $F(a) > w$. Therefore, $x^*$ is a limit point of $h^+_I(w)$, and thus $x^* \in \overline{h^+_I(w)}$. A symmetric argument implies $\overline{h^-_I(w)}\supseteq h^-(w)$ for all $w\in V'$. Then, since $-h(x) + w$ and $h(x) - w$ are both continuous on $W \times \D$, $\overline{h^+_I(w)} \supseteq h^+(w)$, and $\overline{h^-_I(w)} \supseteq h^-(w)$, it holds that $h^+$ and $h^-$ are both l.s.c on $V'$ \cite[Thm 13]{hogan1973point}. Because $h^+$ and $h^-$ are l.s.c, $h^+_I(w)$ and $h^-_I(w)$ are nonempty on $V$, and $h^{-1}(w)$ is assumed to be compact for $w \in [-\delta, \delta]$, $h$ is a l.s.c. point-to-set map on $V \cap V' \cap [-\delta, \delta]$  \cite[Thm 3]{greenberg1972extensions}, \cite{evans1970stability}. Finally, we have $h^{-1}$ is both u.s.c and l.s.c on any open set $U\subset V \cap V' \cap [-\delta, \delta]$ and is therefore continuous on $U$.

Now we prove $\Gamma$ is continuous on $U$.
Since $h$ is assumed to be continuously differentiable, and $f$ in \eqref{eq:sys} is assumed to be continuous, $L_f h$
is also continuous. Because $h^{-1}$ is a continuous point-to-set map on $U$, and $L_f h$ is continuous everywhere,
\begin{equation}
    \Gamma(w) = -\sup \{- L_f h(x) : x \in h^{-1}(w) \}
\end{equation}
is continuous on $U$ \cite[Thm 7]{hogan1973point}. 

Because $\Lambda_\delta$ is assumed to be compact for all $\delta \geq 0$ and $L_f h$ is continuous, $\Gamma$ is bounded on $[-\delta, \delta]$ for all $\delta \geq 0$ as well. Therefore, there exists a continuous function $\phi: \R \to \R$ such that $-\phi(w) \leq \Gamma(w)$ for all $w\in W$ and, for some neighborhood $U'\subset U$ of $0$, $\phi$ restricted to $U'$ is equal to $\Gamma$. Furthermore, $L_f h(x) \geq \Gamma(h(x)) \geq -\phi(h(x))$ for all $x\in \D$.
Since $\S$ is assumed to be invariant, $L_f h(x) \geq 0$ for all $x \in h^{-1}(0)$, so $-\phi(0) = \Gamma(0) \geq 0$. Therefore $\phi$ is continuous and $\phi(0) \leq 0$.
\end{proof}

\end{extendedonly}
\end{document}